\documentclass[pra,preprintnumbers,amsmath,amssymb,superscriptaddress,twocolumn,nofootinbib, showpacs]{revtex4}
\usepackage{bm}
\usepackage{graphicx}
\usepackage{amsbsy}
\usepackage{amsmath}
\usepackage{amsfonts}
\usepackage{bbm}
\usepackage{amsthm}
\usepackage{color}
\usepackage{mathrsfs}
\usepackage{yfonts}
\usepackage{eufrak}
\usepackage{tikz}
\usetikzlibrary{backgrounds}

\begin{document}

\theoremstyle{plain}
\newtheorem{theorem}{Theorem}
\newtheorem{lemma}[theorem]{Lemma}
\newtheorem{conjecture}[theorem]{Conjecture}
\newtheorem{proposition}[theorem]{Proposition}
\newtheorem{corollary}[theorem]{Corollary}
\newtheorem*{remark}{Remark}

\theoremstyle{plain}
\newtheorem{definition}{Definition}
\newtheorem{exmp}{Example}
\newtheorem{example}{Example}[section]

\title{Simulating Symmetric Time Evolution With Local Operations}
\author{Borzu Toloui}
\email{borzumehr@gmail.com}
\affiliation {Institute for Quantum Information Science, University of Calgary, Alberta, T2N 1N4, Canada.}
\author{Gilad Gour}
\email{gour@ucalgary.ca}
\affiliation {Institute for Quantum Information Science, University of Calgary, Alberta, T2N 1N4, Canada.}
\affiliation{Department of Mathematics and Statistics, University of Calgary, Alberta T2N 1N4, Canada.}

\begin{abstract}
In closed systems, dynamical symmetries lead to conservation laws. However, conservation laws are not applicable to open systems that undergo irreversible transformations. More general selection rules are needed to determine whether, given two states,  the transition from one state to the other is possible.
The usual approach to the problem of finding such rules relies heavily on group theory and involves a detailed study of the structure of the respective symmetry group.  
 We approach the problem in a completely new way by using entanglement  to investigate the asymmetry properties of quantum states.   
To this end, we embed  the space state of the system in a tensor product Hilbert space, whereby symmetric transformations between two states are replaced with local operations on their bipartite images.  
 The embedding enables us to use the well-studied theory of entanglement to investigate the consequences of dynamic symmetries.   
 Moreover, under reversible transformations, the entanglement of the bipartite image states becomes a conserved quantity. These entanglement-based conserved quantities are new and  different from the conserved quantities based on expectation values of the Hamiltonian symmetry generators.  Our method is not group-specific and applies to general symmetries associated with any semi-simple Lie group. 
\end{abstract}

\maketitle
\section{Introduction}
\label{sec:intro}

The evolution of most quantum systems is too complicated to be solved analytically or even simulated numerically in an efficient way, at least in the absence of powerful quantum computers~\cite{F82, NBHS11}.
Many  realistic situations involve either open dynamical systems, or closed systems with Hamiltonians that contain numerous  parameters,  so that determining how they vary with time is at the present time not computationally tractable~\cite{Z98, W11}. In all such cases, symmetry-based approaches are powerful substitutes for actual detailed analysis of the complex dynamics involved. Noether's theorem plays a central role in the study of dynamical symmetries of closed systems in classical mechanics~\cite{NT71}. The theorem, as well as its quantum mechanical  counterparts,  state that a Hamiltonian satisfying some  symmetry is always accompanied by a corresponding conservation law~\cite{W54}.  

On the other hand, open systems are more general and  far more ubiquitous than closed systems. Their symmetric time evolutions are expressed by covariant completely positive (CP) maps that in general can be very different from unitary evolutions governed by  symmetry preserving  Hamiltonians.  Hence, no conserved quantities are associated with symmetric dynamics of open systems, and  thus, consequences of dynamical symmetries cannot always be reduced to selection rules based on conservation laws~\cite{MS11a}.  In fact, it was recently shown that even for closed systems, conservation laws given by Noether's theorem do not capture all the consequences of the symmetry in question~\cite{MS11b, MS11c}. It is therefore necessary to look beyond conservation laws in order to  determine how states evolve under symmetric dynamics. 

A system that has a certain symmetry cannot lose the symmetry as it evolves by a Hamiltonian or a master equation that preserve that type of symmetry. More generally, a state cannot become more asymmetric as it undergoes a symmetric time evolution. Comparing the asymmetry properties of two states can therefore be of great use in establishing whether one state can evolve to another under symmetric conditions.  In other  words, when a symmetry is imposed on the dynamics, the \emph{asymmetry} of quantum states becomes a resource~\cite{MS11a}. 
 
In this paper we show that covariant CP-maps can be `simulated' by a restricted subset of local operations and classical communications (LOCC).  
The key idea is to embed the system's Hilbert space~$\mathscr{H}$ within a larger tensor product space~$\mathscr{H}_{A}\otimes \mathscr{H}_{B}$. The embedding is done with an isometry~  
\begin{equation}\label{iso}
\mathscr{H}  \xrightarrow{\text{iso}} \mathscr{W} \subseteq \mathscr{H}_{A}\otimes \mathscr{H}_{B} ,  
\end{equation}
that has the following  properties. 
First, the isometry maps symmetric states to separable states.  
Furthermore, consider two states~$\rho$ and~$\sigma$ that act on~$\mathscr{H}$, and their corresponding bipartite 
image-states~$\tilde{\rho}_{AB}$
and $\tilde{\sigma}_{AB}$ that act on the image subspace~$\mathscr{W}$. If there exists a covariant transformation~$\mathcal{E}_{\text{cov}}$ that maps $\rho$ to $\sigma$, i.e.~$\sigma\equiv\mathcal{E}_{\text{cov}}(\rho)$, then there must also exist a \emph{local}
transformation~$\tilde{\mathcal{E}}_{\text{local}}$ that maps $\tilde{\rho}_{AB}$ to $\tilde{\sigma}_{AB}$, i.e.~$\tilde{\sigma}_{AB}=\tilde{\mathcal{E}}_{\text{local}}(\tilde{\rho}_{AB})$~(Figure~\ref{fig:fig1}). 
 In this sense the local operator~$\tilde{\mathcal{E}}_{\text{local}}$ simulates the covariant map~$\mathcal{E}_{\text{cov}}$.  
  
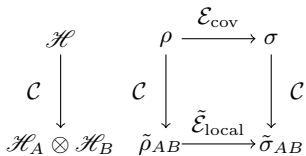
\begin{figure}[h]
 \begin{tikzpicture}[ scale=0.7]
 \node (a) at (0,0) {$\rho$};
 \node (b) at (2,0) {$\sigma$};
 \node (c) at (-0.1,-2) {$\tilde{\rho}_{AB}$};
 \node (d) at (2.2,-2) {$\tilde{\sigma}_{AB}$};
 \node (e) at (-2,0) {$\mathscr{H}$};
 \node (f) at (-2,-2) {$\mathscr{H}_{A} \otimes \mathscr{H}_{B}$};
 \node (1) at (1,0.4) {$\mathcal{E}_{\text{cov}}$};
 \node (2) at (1,-1.6) {$\tilde{\mathcal{E}}_{\text{local}}$};
 \node (3) at (-0.5,-1) {$\mathcal{C}$};
 \node (4) at (2.5,-1) {$\mathcal{C}$};
 \node (5) at (-2.5,-1) {$\mathcal{C}$ };
 \draw[->] (0.3,0)--(1.7,0);
 \draw[->] (0,-0.3)--(0,-1.7);
 \draw[->] (2,-0.3)--(2,-1.7);
 \draw[->] (0.3,-2)--(1.7,-2);
 \draw[->] (-2,-0.3)--(-2, -1.7);
\end{tikzpicture}
 \caption{Simulating a covariant transformation~$\mathcal{E}_{\text{cov}}$ by a LOCC transformation~$\tilde{\mathcal{E}}_{\text{local}}$.}
 \label{fig:fig1} 
\end{figure}

We show here that such isometries can be found for all covariant CP-maps that are associated with compact symmetry Lie groups. Moreover, for any asymmetric state, we show that there exists an isometry that maps it to an entangled state. Hence, the entanglement in the image space capture all the asymmetry properties of the state. 
Our results follow from an application of the Wigner-Eckart theorem, generalized to all semi-simple groups~\cite{CSM10},   that determines the general form of the Kraus operators of covariant transformations~\cite{GS08}.  

The entanglement of the image state plays a somewhat similar role to biomarkers that are employed in biology in order to trace a biological process.  Hence, the study of the evolution of entanglement governed by the LOCC map $\mathcal{E}_{\text{local}}$ opens a new window to explore symmetric dynamics.  
In particular, it shows that the resource theory associated with the asymmetry of quantum states~\cite{MS11a,GS08} is equivalent to the resource theory of entanglement under a restricted subset of LOCC transformations.

 A comprehensive collection of theorems and theoretical tools  has been developed to study  quantum entanglement for more than a decade~\cite{VPRK97, PV07, HHHH09}.  
The equivalence between asymmetry and entanglement resources allows us to take advantage of the repertoire of tools of entanglement theory in order to study the asymmetry properties of quantum states.   
In particular, the established equivalence allows us to use any entanglement monotone and construct a corresponding `asymmetry monotone'~\footnote{In~\cite{MS11a, MS11b} it was called an `asymmetry measure' and in~\cite{GS08} it was called a `frameness monotone'. Here we use the terminology of asymmetry monotones rather than asymmetry measures since these functions do not necessarily \emph{measure} asymmetry, but can sometimes only \emph{detect} it. To see it, consider for example the asymmetry monotone that is equal to 7.2 for asymmetric states, and zero for symmetric states.
Clearly, this monotone does \emph{not} measure asymmetry, only detects asymmetry.}.
An asymmetry monotone, as the name suggests,  is a real function defined on the set of quantum states such that its magnitude  changes monotonically (i.e. non-increasing) during a symmetric evolution.  In the case of reversible symmetric transformations,  asymmetry monotones of course remain conserved. They can thus be regarded as generalizations of conserved quantities.   
Taking asymmetry monotones into account allows us to  rule out  classes of transformations that cannot be ruled out based on conservation laws alone. %
  
A  state that lacks a particular symmetry  encodes information about the physical degrees of freedom that transform under that  symmetry. In contrast, a symmetric state  does not carry any such information.  For example,  the state of electrons with non-zero angular momentum along a particular direction  in space is not symmetric under rotations and consequently encodes some information about that direction, whereas electrons in a rotationally invariant state of zero total angular momentum contain no information about any preferred direction. 

So far,  the study of asymmetry properties of quantum systems has mostly been  focused on pure states. For example, interconversion of pure states under specific symmetry groups has been  studied~\cite{BRS07, GS08, TGS11} and a general classification of pure-state asymmetry properties for arbitrary finite or compact Lie groups has been developed~\cite{MS11b}.  
Prior to the present result,  little was known about the general properties of mixed-state asymmetry, and,  with few but important exceptions like the $G$-asymmetry~\cite{VAWJ08} (also known as the relative entropy of frameness~\cite{GMS09}),  asymmetry monotone functions of mixed-states were not identified for symmetries associated with general groups. 

  Our work  introduces a wide class of asymmetry monotones, defined for all states, pure or mixed. Some of the asymmetry monotones we construct can only be defined in terms of the entanglement of a bipartite system.  A case in point is the negativity measure of entanglement~\cite{VW02}.  Negativity is specially interesting as it provides us with an easily calculable asymmetry monotone for all states and for all types of symmetry. 

Although monotones are extremely useful tools in resource theories~\cite{V00, PV07}, the conditions for the symmetric evolution of states  need not always come in the form of asymmetry monotones. We derive a separate necessary condition for the existence of a covariant transformation from one state to another. However, the condition is such that it cannot be expressed in terms of asymmetry monotones, though for reversible symmetric transformations our necessary condition leads to new conserved quantities.  We arrive at this condition by a new isometry embedding of the system's Hilbert space into a different tensor product structure.  This additional result shows that the isometry in Eq.~(\ref{iso}) can be useful even if it does not simulate covariant transformations with LOCC. 

The paper is organized as follows: in section II, we go over the preliminaries of the asymmetry resource theory, as well as precise definitions of asymmetry monotones. We present our main result in section III  starting with simply reducible compact Lie groups.   In Appendix A, we generalize the main result to general compact Lie groups.  
 Section IV focuses on specific examples of asymmetry monotones and how they compare with their entanglement counterparts.
 In section V, we introduce a new isometry
that in general does not simulate covariant maps with LOCC, but nonetheless leads to new results on time symmetric evolutions.   
Finally we discuss our results and conclusions in section VI. 
Appendix B, contains a special form of the general results for case of  Abelian groups.

\section{ Notations and Preliminaries}
\label{sec:prem}
In this section, we briefly discuss few key elements of the resource theory of quantum asymmetry that we will be using in the rest of the paper.  
In particular, we go over G-covariant maps, irreducible tensor operators, the Wigner-Eckart theorem, and asymmetry monotones.
For a more detailed review on asymmetry, and its relation to reference frames, and super-selection rules see~\cite{MS11a,GS08,BRS07}. 
   
\subsection{G-covariant Transformations}
\label{subsec:G-cov}
 Let~$\mathcal{B}(\mathscr{H})$ denotes the set of bounded operators over~$\mathscr{H}$. 
Consider a completely positive (CP) map~$\mathcal{E}: \mathcal{B}(\mathscr{H}) \rightarrow \mathcal{ B }(\mathscr{H})$ that takes density matrices to density matrices. Let~$G$ be a group of transformations, and define the map~$\mathcal{U}(g)$  as
\begin{align}
\label{def:U}
   \mathcal{U}(g)[\bullet]:=U(g) (\bullet) U^\dagger(g),
\end{align}      
 where~$U: G \rightarrow \mathcal{B}(\mathscr{H}): g\mapsto U(g)$ is a representation of the group~$G$. In this paper we only consider   compact semi-simple Lie   groups with fully reducible unitary representations\footnote{Our method relies heavily on the general form of the Wigner-Eckart  theorem for semi-simple Lie groups. A much more complicated form of the Wigner-Eckart theorem exists for finite groups, sometimes known as  the Koster-Wigner-Eckart theorem~\cite{Koster58}. It might still be possible that this finite counterpart of the Wigner-Eckart theorem can lead to results analogous to ours, but we do not consider those cases in this paper.}.    A semi-simple Lie group is a Lie group whose algebra is semi-simple~\cite{W74}  . We will also assume that the representation of the group comes with a group independent (Haar) measure. 
 Compact Lie groups all have unitary representations with Haar measures.

 We say that the mapping~$\mathcal{E}$ is symmetric with respect to~$G$, or equivalently, that~$\mathcal{E}$ is $G$-covariant,  if  for all~$\rho$ and for all~$g \in G$, 
 \begin{align}
 \label{def:cov}
 \mathcal{E}\circ\mathcal{U}(g) [\rho] = \mathcal{U}(g)\circ\mathcal{E} [\rho]. 
 \end{align}
 In particular, if the CP-map consists of a single unitary~$\mathcal{E}[\bullet]=V (\bullet) V^\dagger$, then the condition of~$G$-covariance in Eq.~(\ref{def:cov}) becomes
 \begin{align}
 [U(g),V]=0, \; \forall g \in G. 
 \end{align}
 The unitary~$V$ is called~$G$-invariant in this case.     
 Similarly, a symmetric state~$\rho$ is any state that remains invariant under the application of the group representation, also known as a $G$-invariant state, 
 \begin{align}
 [U(g), \rho]=0, \; \forall g \in G,  
 \end{align} 
 which is equivalent to 
 \begin{align}
  \forall g \in G,~U(g)\rho \:U^{\dag}(g)=\rho. 
\end{align}
 Consider the uniform average of the group action:
  \begin{align}
 \label{def:twirl}
 \mathcal{G} [\rho]:=\int \text{d}\mu_g\; \mathcal{U}(g)[\rho] ;, 
 \end{align}
where~$\text{d}\mu_g$ denotes the group invariant (Haar) measure. In the case of discrete groups the integral is replaced with a sum, and the uniform measure with the inverse group size. The averaging superoperator~$\mathcal{G}$ in Eq.~(\ref{def:twirl}) is known as the~$G$-twirling operation.  It follows from the uniformity of the measure that twirled states are invariant under the action of any element of the group,~i.\ e.\ they are~$G$-invariant. In fact, it can be shown that every~$G$-invariant state can be expressed as the outcome of a twirling operation~\cite{BRS07}. 
  
\subsection{Irreducible Tensor Operators And The Wigner-Eckart Theorem}
\label{subsec:ito}
 
Let the irreps of the group~$G$ be labeled by the letter~$j$. In general~$j$ can be a short hand notation for a multiple of independent labels that together fully label the irreps.   As the irreps are unitary,~$j$ denotes the highest weight of each irrep and is a  vector of dimension~$\ell$,  where~$\ell$ is the rank of the associated algebra. Also let~$m$ label the basis vectors of the irrep, i.\ e.\ the basis vectors spanning the invariant subspace of the~$j$'th irrep. In fact, ~each $m$ denotes a weight of the irrep labeled by~$j$ and is thus, also a~$\ell$-dimensional vector in the weight space of the irreps.  

Also, let us decompose the Hilbert space as
\begin{align}
\label{eq:Hj}
\mathscr{H}=\bigoplus_{j, \lambda} \mathscr{H}_{j, \lambda}
\end{align}
where~$\mathscr{H}_{j, \lambda}$ carries an irrep labeled by~$j$, where  the range of~$j$ is assumed to be unbounded. The index~$\lambda$ labels the multiplicity of the irrep.   
With this decomposition of $\mathscr{H}$,
the~$G$-twirling of a state~$\rho$ has the form
\begin{align}
\label{eq:twirl2}
\mathcal{G}[\rho]= \sum_{j, \lambda} p_{j, \lambda} \Pi_{j, \lambda},
\end{align} 
 where~$\Pi_{j, \lambda}$ is the projection onto subspace~$\mathscr{H}_{j, \lambda}$ that carries the~$j$th irrep.

The definition of~$G$-covariance in Eq.~(\ref{def:cov}) is equivalent to 
\begin{equation}\label{gcov}
\mathcal{E}=\mathcal{U}(g)\circ \mathcal{E}\circ \mathcal{U}(g^{-1}), \;\; \forall g \in G.
\end{equation}
Clearly, if~$\left\{K_i\right\}$ is a set of Kraus operators of a~$G$-covariant CP-map $\mathcal{E}$, then, from Eq.~(\ref{gcov}), it follows that~$\left\{U(g) K_i U^\dagger(g)\right\}$ is also a set of Kraus operators for $\mathcal{E}$.   
Now, two operator sum representations of the same channel $\mathcal{E}$ are related by a unitary matrix. Therefore, it follows that  
\begin{align}
U(g) K_i U^\dagger(g) =\sum_{i'} u_{ii'}(g)\; K_{i'} 
\end{align}
where~$u_{ii'}(g)$ are the elements of a unitary matrix~$u(g)$. It was shown in~\cite{GS08} that if the~$\left\{K_i\right\}$ are linearly independent, then $u(g)$ is also a representation of the group~$G$. Furthermore, bringing the matrix~$u(g)$ to the block diagonal form,  ~
\begin{align}
u(g)=\bigoplus_{j, \lambda} u_{j, \lambda}(g)
\end{align}
    of the group's irreps,   simply amounts to a different unitary remixing of the Kraus operators, and is thus allowed. This, in turn, means that the Kraus operators of a~$G$-covariant CP-map can be grouped into subsets that mix only among themselves, each labeled by the irrep labels of the group.

Thus, every~$G$-covariant CP-map admits a Kraus decomposition labeled~$K_{j,m,\alpha}$, with~$\alpha$ being a multiplicity index, such that 
\begin{align}
\label{def:irtenop}
K_{j,m,\alpha}= \sum_{m'} u^{(j)}_{m,m'}(g)\; K_{j,m',\alpha}, \;\; \forall g\in G. 
\end{align} 
For each irrep label~$j$, Kraus operators of the set~$\left\{K_{j,m,\alpha}\right\}$ are called irreducible tensor operators of rank~$j$. 

A CP-map with a Kraus decomposition comprised of a set of irreducible tensor operators,  
\begin{align}
\label{def:Ej}
\mathcal{E}_{j,\alpha} (\bullet)= \sum_m K_{j,m,\alpha} \; (\bullet) \; K^\dagger_{j,m,\alpha}, 
\end{align}
is an irreducible~$G$-covariant operation. Every~$G$-covariant CP-map can be expressed as a sum of irreducible~$G$-covariant operations. 

\subsubsection{ The Wigner-Eckart Theorem }
\label{subsubsec:WE}
The Wigner-Eckart theorem determines the matrix elements of the irreducible tensor operators with respect to the~$SU(2)$ algebra, also known as spherical tensor operators~(for example see pp. 193-195 in~\cite{B98}).  In fact, the Wigner-Eckart theorem can be generalized and applied to any compact, semi-simple group and its associated Lie algebra~\cite{CSM10}. For simplicity of the exposition, we will first assume that the 
Kronecker product of the algebra associated with the group is simply reducible. That is, the coupling of two irreps has no outer multiplicities (i.e. multiplicities that arise due to coupling). We leave the generalization to all semi-simple compact groups to Appendix~\ref{A1}.   The Wigner-Eckart theorem then specifies the form of the  matrix elements of~$K_{J,M,\alpha}$ as we now discuss.    

Let ~$\left\{|j, \lambda ; m  \rangle \right\}$~  be the set of basis vectors spanning the Hilbert space $\mathscr{H}$.
Here~$m$ labels   the weights of    
the~$j$'th irrep,   as before,  and~$\lambda$  labels the multiplicity of the irrep. 
The Wigner-Eckart theorem states that the matrix elements of~$K_{j,m,\alpha}$ are given by: 
\begin{align}
\label{eq:WE}
\langle  j',\lambda'; m' |K_{J,M,\alpha}& |j, \lambda; m \rangle= \nonumber \\ 
&\left(
\begin{matrix}
  j & J \\
  m & M
 \end{matrix}
\right |
\left. 
 \begin{matrix}
 j'  \\
  m'
 \end{matrix}
 \right)
\: \langle j', \lambda' \parallel K_{J, \alpha} \parallel j, \lambda\rangle, 
\end{align}
where~$\langle j', \lambda' \parallel K_{J, \alpha} \parallel j, \lambda\rangle$ is the reduced matrix element independent of~$m$ and~$m'$, and~$\left(
\begin{matrix}
  j & J \\
  m & M
 \end{matrix}
\right |
\left. 
 \begin{matrix}
 j'  \\
  m'
 \end{matrix}
 \right)
$ are the (general) Clebsch-Gordan (CG) coefficients. 
  
\subsubsection{Clebsch-Gordan Coefficients}
\label{subsubsec:cg}
The generalized CG coupling coefficients~$
\left(
\begin{matrix}
  j_1 & j_2 \\
  m_1 & m_2
 \end{matrix}
\right |
\left. 
 \begin{matrix}
 j_3 \\
  m_3
 \end{matrix}
 \right)
$ relate the basis~$|j_1,m_1\rangle \otimes |j_2, m_2\rangle$  to the basis~$\left|j_3, m_3 ; (j_1, j_2)\right\rangle$  that reduces the Kronecker product of the two irreps,  
\begin{align}
\left|j_3, m_3; \right. &\left (j_1, j_2)\right\rangle=\nonumber \\
&\sum_{m_1, m_2}
\left(
\begin{matrix}
  j_1 & j_2 \\
  m_1 & m_2
 \end{matrix}
\right |
\left. 
 \begin{matrix}
 j_3 \\
  m_3
 \end{matrix}
 \right)
\; |j_1,m_1\rangle \otimes |j_2, m_2\rangle. 
\end{align}
 Here, we have dropped the multiplicity index~$\lambda$, as the CG-coefficients do not depend on the multiplicity. In the rest of the paper, we use~$|j; m\rangle$, or~$|j, \lambda; m\rangle$   instead of~$\left|j, m ; (j_1, j_2)\right\rangle$ or~$\left|j, \lambda, m ; (j_1, j_2)\right\rangle$ for brevity whenever the context  is clear.

\subsection{Monotones}
\label{subsec:mon}
Every restriction on quantum operations defines a resource theory, determining how quantum states that cannot be prepared under the restriction may be manipulated and used to circumvent the restriction. Here we discuss briefly how the resourcefulness of these quantum states is quantified.
We will focus on entanglement theory and the theory of asymmetry that is associated with a group~$G$ of transformations.
In entanglement theory, the quantum operations or CP-maps are confined to LOCC, and only separable states can be prepared by LOCC (assuming no access to previously existing entanglement).   In the resource theory of quantum asymmetry, the only allowed operations are~$G$-covariant CP maps, and the only states that can be prepared without any resources are~$G$-invariant states. 

A quantum state cannot turn into a stronger resource by the set of restricted (or allowed) operations. Therefore, the strength of the resource 
must be quantified by functions that do not increase under the set of allowed operations. Such functions are called monotones.
We now give a precise definition for monotones that we will use in the rest of the paper, and that apply to both entanglement and asymmetry. 

The most general quantum transformation converts an initial state $\rho$ into one of a set of possible final states, say $\sigma_{x}$, that occurs with probability $p_x$.
Such a general quantum transformation is described by a CP map~$\mathcal{E}: \mathcal{B}(\mathscr{H})\rightarrow \mathcal{B}(\mathscr{H})$ that is itself composed of a number of CP (in general trace decreasing) maps~$\{\mathcal{E}_x\}$, so that $\mathcal{E}=\sum_x\mathcal{E}_x$,  and
\begin{equation}
\label{eq:sigmax}
	\sigma_x:=\mathcal{E}_x[\rho]/p_x, 
\end{equation}
where the probability~$p_x=\text{Tr}\left(\mathcal{E}_x[\rho]\right)$. We say that $\mathcal{E}$ is $G$-symmetric if all
$\{\mathcal{E}_x\}$ are $G$-covariant.  

The ensemble of outcomes is written as $\left\{\sigma_x, p_x\right\}$. This ensemble can be equivalently expressed as a  density operator 
\begin{align}
\tilde{\sigma}:= \sum_x p_x\;\sigma_x \otimes |x\rangle\langle x|, 
\end{align}
 where $\{|x\rangle\}$ are a set of mutually orthogonal unit states. ~$\tilde{\sigma}$ can be prepared out of the ensemble~$\left\{\sigma_x, p_x\right\}$ by annexing an ancilla in the state labeled by the index~$x$, and then forgetting the value of~$x$. Reversely, the ensemble can always be reproduced from the density operator~$\tilde{\sigma}$ by performing the  measurement~$\mathcal{M}=\{|x\rangle\langle x|\}_x$.  
 \begin{definition}
\label{def:monotones}
Using the above notations, a function $A: \mathcal{B} \left(\mathscr{H}\right)\to\mathbb{R}^{+}$ is called an asymmetry (entanglement) monotone, 
if it satisfies
\begin{align}\label{mon}
 A(\rho)\geq A(\tilde{\sigma})  
\end{align}  
for all CP maps $\mathcal{E}$ that are $G$-  covariant  (LOCC).
\end{definition}

We can further classify the asymmetry (entanglement) monotones into another category: 

 \begin{definition}
\label{def:ensemble}
The asymmetry (entanglement) monotone $A: \mathcal{B} \left(\mathscr{H}\right)\to\mathbb{R}^{+}$ is called an \emph{ensemble} monotone
if it satisfies
\begin{align}
A(\rho)\geq\sum_x p_x A(\sigma_x),
\end{align}
for all CP maps $\mathcal{E}$ that are $G$-symmetric (LOCC).
\end{definition}
Note that the set of ensemble monotones is a strict subset of the monotones defined in Eq.~(\ref{mon}).
The most restrictive monotones are monotones that do not increase under any non-deterministic (trace-non-increasing) CP-map~$\mathcal{E}_x$.

\begin{definition}
\label{def:faithful}
An asymmetry (entanglement) monotone~$A: \mathcal{B} \left(\mathscr{H}\right)\to\mathbb{R}^{+}$ is faithful if
\begin{align}
A(\rho)= 0\; \text{ iff }\; \rho\; \text{is~$G$-invariant (separable).} 
\end{align}
\end{definition}

 We are now ready to present our main result that connects~$G$-covariant transformations to LOCC transformations and entanglement monotones to asymmetry monotones.

\section{Simulating~$G$-covariant transformations}
\label{sec:sim}
As discussed in section~\ref{sec:intro},  the central idea of the present paper is to embed the system's Hilbert space within a larger Hilbert space in such a way that the covariant transformations between original states map to LOCC transformations in the larger Hilbert space. We now proceed to make precise the concepts and procedures  involved. 
We use the notations introduced in sections~\ref{sec:intro} and~\ref{sec:prem}. 
\begin{definition}
\label{def:lociso}
A LOCC-simulating isometry is an isometry~$\mathcal{C}: \mathcal{B}\left(\mathscr{H}\right) \rightarrow \mathcal{B}\left(\mathscr{W}\right)$, with a bipartite image space $\mathscr{W}\subseteq \mathscr{H}_{A} \otimes \mathscr{H}_{B}$ (see Eq.~\ref{iso}),
that satisfies the following three conditions:\\
\textbf{(1)} For any $G$-covariant map, $\mathcal{E}_{\text{cov}}$, the map
$$\mathcal{C} \circ \mathcal{E}_{\text{cov}} \circ \mathcal{C}^{-1}\equiv\mathcal{E}_{\text{local}}:\;\mathcal{B}\left(\mathscr{W}\right)\to\mathcal{B}\left(\mathscr{W}\right)$$ is local; that is, $\mathcal{E}_{\text{local}}$ can be implemented by LOCC. \\
\textbf{(2)} If~$\rho$ is~$G$-invariant then~$\mathcal{C}(\rho)$ is separable. \\
\textbf{(3)} There exists an asymmetric state (i.e. non-$G$-invariant state)~$\sigma$ for which~$\mathcal{C}(\sigma)$  is entangled. 
\end{definition}

The third point excludes trivial isometries that map every state, whether $G$-invariant or not,  to a separable state. One example of such a trivial isometry is simply adding an ancilla state to every state~$\rho$, i.\ e., 
\begin{align}
\rho \mapsto \mathcal{C}(\rho):= \rho \otimes |0\rangle \langle 0|. 
\end{align}

Trivial isometries of this sort are of course always possible, but they differentiate neither between~$G$-invariant and non-invariant states,  nor between~$G$-covariant or non-covariant transformations.   
Thus, they tell us nothing about the states' asymmetry properties or about the conditions under which covariant transformations are possible.  
The other extreme, that of mapping every asymmetric state to an entangled state, although ideal, is not likely to always be possible.  The isometries that we consider here do not fall under either extreme. Nevertheless, we are able to
find a \emph{set} of isometries that is complete in the sense that for any asymmetric state there exists at least one isometry in the set that takes it to an entangled state.  In this sense, entanglement capture all aspects of asymmetry.
 
 \subsection{The Main Isometry}
 
The Wigner-Eckart theorem determines the matrix elements of an irreducible tensor operator, like the Kraus operators of~$G$-covariant transformations,  in the basis ~$|j,\lambda; m \rangle$ introduced in section~\ref{subsec:G-cov}~(see Eq.~(\ref{eq:WE})).   
An important consequence of the Wigner-Eckart theorem  is the existence of the so called selection rules.  The generalized CG coupling coefficients~$(j, m; J, M| j', m')$  are zero unless the weights~$m$,~$M$ and~$m'$ satisfy the relation,   
\begin{align}
\label{eq:m}
m+M=m'. 
\end{align} 
The matrix elements that do not satisfy Eq.~(\ref{eq:m}) must vanish.  It thus follows from  the Wigner-Eckart theorem that the only thing a~$G$-covariant Kraus operator~$K_{J,M,\alpha}$ does on the weight~$m$ of a basis state is to translate it  by~$M$, independently of the other relevant parameters,~$j$,~$J$,~$\lambda$ and~$\alpha$.  
We exploit this fact in the following definition and theorem when we introduce an isometry that satisfies the three conditions of definition~\ref{def:lociso}. 
 
  \begin{definition}
 \label{def:iso1}
 Let~$\mathscr{H}_B$ denote the Hilbert space spanned by kets~$|m\rangle$ where~$m$ ranges over the representation weights of the associated algebra of the group, and let    
 \begin{align}
 \label{def:W}
 \mathscr{W}:=\text{span} \left\{|j,\lambda; m\rangle \otimes  |m\rangle\right\}\subset \mathscr{H}\otimes{\mathscr{H}_B}.    
\end{align}
The isometry~$\mathcal{C}$~is defined by its action on the basis kets as:
 \begin{align}
|j, \lambda; m\rangle \xrightarrow{\mathcal{C}}  |j,\lambda; m\rangle \otimes |m\rangle. 
 \end{align}
 \end{definition}  
We now show that $\mathcal{C}$ satisfies all the conditions of definition \ref{def:lociso}.  
\begin{proposition}
\label{prop:C1}
$\mathcal{C}$ is a LOCC-simulating isometry. 
\end{proposition}
\begin{proof}
To see that the first condition in definition~(\ref{def:lociso})  is satisfied, consider a~$G$-covariant CP-map~$\mathcal{E}_{\text{cov}}$ whose operator sum representation is given in terms of Kraus operators~$\left\{K_{J,M, \alpha} \right\}$. We define,     
\begin{equation}\label{osr}
\tilde{K}_{J,M,\alpha}:= K_{J,M,\alpha} \otimes T_M, 
\end{equation}
where, 
\begin{align}
\label{def:T}
T_M:= \sum_{m} |m+M\rangle \langle m|, 
\end{align}
is a translation operator. 
Let~$\mathcal{E}_{\text{local}}$ denotes the CP-map whose operator sum representation corresponds to the Kraus operators~$\tilde{K}_{J,M,\alpha}$ given in Eq.(\ref{osr}).  Note that from Eq.~(\ref{eq:m}) it follows that $\mathcal{E}_{\text{local}}=\mathcal{C} \circ \mathcal{E}_{\text{cov}} \circ \mathcal{C}^{-1}$.
We need to show that $\mathcal{E}_{\text{local}}$ can be implemented by LOCC. Indeed,
note that~$T_M$, being merely a translation operator, is unitary (assuming the range of the weights in the decomposition~(\ref{eq:Hj}) is unbounded).  Therefore, the map $\mathcal{E}_{\text{local}}$ can be implemented as follows: Alice perform a `local' measurement described by the Kraus operators $\{K_{J,M,\alpha}\}$ and send the part $M$ of the her measurement outcome to Bob,
who then perform the unitary transformation $T_{M}$. Hence, the first criterion of Definition~\ref{def:lociso} is satisfied.


Secondly, recall that any~$G$-invariant state~$\rho$ is equal to its own~$G$-twirling~(see Eq.~\ref{eq:twirl2}), 
\begin{align}
\rho= \sum_{j, \lambda} p_{j, \lambda} \Pi_{j, \lambda}, 
\end{align}
where the projection ~$\Pi_{j, \lambda}$ is equal to
\begin{align}
\Pi_{j, \lambda}=\sum_m |j,\lambda; m\rangle \langle j,\lambda; m|. 
\end{align}   
The state~$\mathcal{C}(\rho)$ is thus equal to 
\begin{align}
\mathcal{C}(\rho)=\sum_{j, \lambda} p_{j, \lambda} \sum_m |j,\lambda; m\rangle \langle j,\lambda; m|\otimes |m\rangle\langle m|, 
\end{align}
which is clearly a separable state. 

Finally, a state of the form 
\begin{align}
\label{ex:psi}
|\psi\rangle = c_1\: |j_1, \lambda_1; m_1\rangle + c_2\: |j_2, \lambda_2; m_2\rangle,   
\end{align} 
is mapped to the entangled state, 
\begin{align}
|\tilde{\psi}\rangle=c_1\: |j_1, \lambda_1; m_1\rangle\otimes |m_1\rangle + c_2\: |j_2, \lambda_2; m_2\rangle\otimes |m_2\rangle.   
\end{align}     
This completes the proof.  
\end{proof}
 
The example in Eq.~(\ref{ex:psi}) suggests that if a state has coherence in~$m$ it is mapped to an entangled state. In the next proposition, we make this claim rigorous and give necessary and sufficient conditions for a general mixed state~$\rho$ to be mapped to an entangled state by the isometry~$\mathcal{C}$. 
\begin{proposition}
\label{prop:gen}
Let $\Pi_m$ be the projection
$$
\Pi_{m}:=\sum_{j,\lambda}| j, \lambda; m\rangle\langle j, \lambda; m|\;. 
$$ Then,
the isometry~$\mathcal{C}$ maps a state~$\rho$ to an entangled state if and only if there exists $m$ such that~$[\rho,\Pi_m]\neq 0$; i.e. $\rho$ has coherence in~$m$. 
\end{proposition}

\begin{proof}
Every state~$\tilde{\rho}$ acting on~$\mathscr{W}$ is the image of some state acting on~$\mathscr{H}$, i.\ e.\ ~$\tilde{\rho}=\mathcal{C}(\rho)$.  If~$\tilde{\rho}$ is a separable state, it must have a pure state decomposition comprised of product states~
$$
\tilde{\rho}=\sum_i q_i |\tilde{\phi}_i\rangle\langle \tilde{\phi}_i|,
$$
 where each~$|\tilde{\phi}_i\rangle$ is both a product state and the image of some state~$|\phi_i\rangle$ under the isometry~$\mathcal{C}$. This is because $\mathcal{C}$ is a linear invertible map, and any pure state decomposition of $\rho$ corresponds to a pure state decomposition of $\tilde{\rho}$ and vice versa.
 Thus, since $|\tilde{\phi}_i\rangle=\mathcal{C}(|\phi_i\rangle)$ is a product state, $|\phi_i\rangle$ must have the form
 \begin{align}
\label{phi}
|\phi_i\rangle=\sum_{j, \lambda} c_{i; j, \lambda} |j, \lambda; m_i\rangle.    
\end{align}
where $c_{i; j, \lambda}$ are some complex coefficients and the superposition above consists of a \emph{single} value for~$m=m_i$. Otherwise, containing two different values for~$m$ in the above expansion necessarily renders the state $|\tilde{\phi}_i\rangle$ entangled. Consequently, the form of the initial state~$\rho$ corresponding to~$\tilde{\rho}=\mathcal{C}(\rho)$ must be,  
$$
\rho=\sum_i q_i |\phi_i\rangle\langle\phi_i|\;,
$$
with $|\phi_i\rangle$ as in~Eq~\eqref{phi}.
According to Eq.(\ref{phi}) each $|\phi_i\rangle\langle\phi_i|$ commutes with $\Pi_m$ for all $m$ and therefore $[\rho,\Pi_m]= 0$. 
The argument works in the other direction as well. In other words,~if every pure state decomposition of~$\rho$ contains at least one pure state that is in a coherent superposition of two or more eigenstates with different values of~$m$, then~$\mathcal{C}(\rho)$ will be an entangled state.  
This completes the proof.  
\end{proof}
 
 The isometry $\mathcal{C}$ does not map \emph{all} asymmetric states to entangled states.  For example, the state
$
|\phi\rangle= |j, \lambda; m\rangle  
$ 
is not~$G$-invariant~(assuming~$G$ is non-Abelian and~$j$ does not label the identity irrep), and yet it is mapped to the product state  
$$
|\tilde{\phi}\rangle = |j, \lambda; m\rangle \otimes |m\rangle.  
$$
However, as we now show, we can define another LOCC-simulating isometry, similar to $\mathcal{C}$, that maps
$|j, \lambda; m\rangle$ to an entangled state.

\subsection{A complete set of LOCC-simulating isometries}
\label{subsec:Cg}

In Definition~\ref{def:iso1} we have used the basis $\{|j, \lambda; m\rangle\}$  to define the isometry $\mathcal{C}$.
However, there is nothing special about the choice of the irrep weights~$m$. In fact, the set of states,
\begin{align}
\label{mgm}
|j, \lambda; m\rangle_g := U(g) |j, \lambda; m\rangle 
\end{align}  
forms an equally valid basis for the irreps, labelled by the new weights~$m_g$ (the multiplicity index~$\lambda$ can always be relabelled if it is needed). On the other hand, by definition, the irrep basis states mix among themselves under the action of the group, 
\begin{align}
\label{eq:ketg}
U(g) |j, \lambda; m\rangle = \sum_{m'} D^{(j)}_{m, m'}(g) \:|j, \lambda; m'\rangle,  
\end{align} 
where~$D^{(j)}_{m, m'}(g)$ is the matrix representation of the~$j$'th  irrep.  
Reversing Eq.~(\ref{mgm}), we get, 
\begin{align}
\label{mmg}
|j, \lambda; m\rangle = \sum_{m'} D^{(j)}_{m, m'}(g^{-1}) \:|j, \lambda; m'\rangle_g. 
\end{align} 

Hence, if we had defined the isometry relative to the new weights, the state~$|j, \lambda; m\rangle$ would be mapped to an entangled state. In fact, the isometry~$\mathcal{C}$ is only one of a class of isometries that can be defined for different choices of~$g \in G$ relative to the  weights~$\left\{m_g\right\}$.~$\mathcal{C}$ is merely the isometry corresponding to the identity element of the group.   
  
\begin{definition}
 \label{def:isog}
For every~$g \in G$,  we define the isometry ~$\mathcal{C}_g$ as,  
 \begin{align}
 \mathcal{C}_g:= \left(\mathcal{U}(g) \otimes \mathcal{I}_B\right) \circ \mathcal{C} \circ \mathcal{U}^{\dagger}(g),  
 \end{align}  
 where~$\mathcal{U}(g):=U(g)(\bullet) U^{\dagger}(g) $, 
 and ~$\mathcal{I}_B$ is the identity superoperator acting on~$\mathscr{H}_B$.   
 \end{definition} 
 The isometry~$\mathcal{C}_g$ acts on basis states,~$|j, \lambda; m\rangle_g$, and maps them to 
 $$
 |j,\lambda; m\rangle_g \mapsto |j,\lambda; m\rangle_g \otimes |m\rangle.  
 $$  
  
Note that the image space of all the isometries~$\{\mathcal{C}_g\}$ is the space~$\mathscr{W}$ in~(\ref{def:W}).  Clearly, the proof of Proposition~\ref{prop:C1} can easily be modified to apply to all the set of isometries $\{\mathcal{C}_g\}$. 
Moreover, note that  the state $|\phi\rangle\equiv|j, \lambda; m\rangle$ is mapped to
\begin{align}
|\phi\rangle \xrightarrow{\mathcal{C}_g} |\tilde{\phi}\rangle=  \sum_{m'} D^{(j, \lambda)}_{m, m'}(g^{-1}) \:|j, \lambda; m'\rangle_g \otimes |m'\rangle, 
\end{align}
which is, in general, an entangled state. 


It is instructive at this stage to look at the specific group of rotations~SO(3), or similarly, the group~SU(2) to gain some intuition. The weights~$m$ of the associated algebra~$\mathfrak{su}(2)$ are one dimensional and correspond to the eigenvalues of the angular momentum operator, ~$J_z$,  along the~$z$-direction. Each irrep is labeled by the single number~$j$ corresponding to  the maximum~$z$-eigenvalue of angular momentum, and the total angular momentum,~$J^2$, equals~$j(j+1)$.  There is obviously nothing special about the choice of the~$z$-axis. The~$z$-axis can be rotated to a new  axis~$\hat{\text{n}}$, which corresponds to applying the respective group representation on the quantum states. 
One way to specify an element of the  group is to determine the axis~$\hat{\text{n}}$  to which it takes the~$z$-axis.  In other words, each isometry in the class of definition~\ref{def:isog} is identified by the choice of a new~$z$-direction and can be denoted as~$\mathcal{C}_{\hat{\text{n}}}$. 

Thus, to take full advantage of the entanglement features of the embedding, one has to take more than one isometry into consideration.     
 As we shall now see, if $\rho\in\mathcal{B}\left(\mathscr{H}\right)$ is an asymmetric state then there exists 
$g\in G$ such that $\mathcal{C}_g(\rho)$ is an entangled state. In fact, for  the $SU(2)$ group
we will see that only two directions are needed to characterize all the asymmetry properties of a state. That is, if 
$\mathcal{C}_{\hat{\text{n}}}(\rho)$ is separable for two independent choices of $\hat{\text{n}}$, then $\rho$ is necessarily
$G$-invariant. 

Also for more general simply connected groups,  there exists a \emph{finite} number of isometries $\{\mathcal{C}_{g_i}\}$ (associated with a finite number of group elements~$\{g_i\}$) that capture all the asymmetry properties of a state.  That is, if a state is mapped to a separable state by all the isometries in the finite set $\{\mathcal{C}_{g_i}\}$, then the state must be symmetric.  This allows in principle to check whether a state is~$G$-invariant or not, by considering  its bipartite image states only for a finite number of isometry elements. Otherwise, all the infinite isometries, each associated with a member of the group, must have been considered before such an assessment could be made.

Before proving the above claim rigorously, let us illustrate the idea of the proof with the simple and more familiar example of the group~SU(2).  Suppose that~$\mathcal{C}(\rho)$ is separable for some state~$\rho$. Then, according to Proposition~\ref{prop:gen}
the state~$\rho$ has no coherence in~$m$, the eigenvalue of the~$J_z$ operator. It means, in turn, that the state $\rho$ commutes with~$J_z$. By the same argument if~$\mathcal{C}_{\hat{x}}(\rho)$ is separable then the state $\rho$ commutes with~$J_x$. 
Therefore, if both~$\mathcal{C}(\rho)$ and~$\mathcal{C}_{\hat{x}}(\rho)$ are separable then $\rho$ commutes with both $J_z$
and $J_x$. But since $[J_z, J_x]=\imath J_y$, $\rho$ also commutes with $J_y$ and so it
must commute with all the elements of the group which means that $\rho$ is an $SU(2)$-invariant state. This line of argument can be generalized to other groups, as we now demonstrate. 

Suppose~$G$ is a simply connected group parametrized by~$r$ parameters. Let~$\mathfrak{g}$ be the associated algebra of~$G$ of rank~$\ell$, and let~$\mathfrak{h}$ be its ~$\ell$-dimensional Cartan subalgebra.  Denote the operator representation of  the infinitesimal generators of the group as ~$X_{a}: \mathscr{H} \rightarrow \mathscr{H}$, for~$a=1\cdots r$. Similarly, denote the representation of the Cartan operators spanning~$\mathfrak{h}$ as~$H_i: \mathscr{H} \rightarrow \mathscr{H}$, where~$i=1\cdots \ell$.  

Now, let~$S \subset G$ be the subgroup of~$G$ whose members permute the infinitesimal generators of the group among themselves. By this we mean, for every~$s \in S$,       
\begin{align}
\label{eq:S}
U(s)\: X_{a}\: U(s)^{\dag} = X_{a'(s)}, \:\: a, a' \in \left\{1 \cdots r\right\}. 
\end{align}
 As both~$\mathfrak{g}$ and~$\mathfrak{h}$ are finite, the subgroup~$S$ contains only a \emph{finite} number of elements.  
 We are now ready to prove the general case.  
\begin{proposition}
\label{prop:covent}
Let $\rho\in\mathcal{B}\left(\mathscr{H}\right)$. Then, $\rho$ is $G$-invariant if and only if  
 for all $s$ belonging to the finite subgroup~$S\subset G$,  the state $\mathcal{C}_s(\rho)$ is separable.  
\end{proposition}

\begin{proof}
If $\rho$ is $G$-invariant, then $\mathcal{C}_g(\rho)$ is separable for all~$g \in G$, and thus for all~$s \in S$,  since $\left\{\mathcal{C}_g\right\}$ is a set of LOCC-simulating maps. 

We therefore assume that $\mathcal{C}_s(\rho)$ is separable for all $s\in S$.  The requirement that  $\mathcal{C}_s(\rho)$ is separable implies that $\rho$,  when expressed in the basis~$|j, \lambda; m\rangle_s$,  has no coherence in~$m$. 
Consider  the projection, 
$$
\Pi^{(s)}_{m}:=\sum_{j,\lambda}| j, \lambda; m\rangle_s\langle j, \lambda; m|.
$$
The condition for separability is equivalent to the requirement that~ $[\rho,\Pi^{(s)}_{m }]=0$ for all~$m$ (see Proposition~\ref{prop:gen}). 

The set of operators, ~
$
H^{(s)}_{i}:=U(s)\: H_i \:U(s)^{\dag},  
$
 are all diagonal in the new basis,    
$$
H^{(s)}_{i} \:|j, \lambda; m\rangle_s = m_i \:|j, \lambda; m\rangle_s,     
$$ 
and form a representation for new Cartan operators. 
It follows that~  
$
H^{(s)}_{i}= \sum_m m_i \:\Pi^{(s)}_{m}.     
$
 Thus, if~$\mathcal{C}_s(\rho)$ is separable,~$\rho$ must satisfy    
$$
[\:\rho, H^{(s)}_{i}]=0, \;\; i=1\cdots \ell. 
$$
But this is true for all~$s \in S$ (including the identity~$e$, where~$H_i\equiv H^{(e)}_{i}$).  Every~$X_a$ can be constructed from the commutators of~$H^{(s)}_{i}$, once all the~$H^{(s)}_{i}$ for all~$s \in S$ are included.    
It follows that the state~$\rho$  commutes with all the generators~$X_a$, and consequently, with all the elements of the group as well.  In other words, the state is~$G$-invariant.  

 \end{proof}

In the next section,  we see how entanglement of the embedded state changes under~$G$-covariant transformations of the original state. This, in turn, enables us to relate the asymmetry features of the original state to the ensuing entanglement.

\subsection{Constructing Asymmetry Monotones From Entanglement Monotones}
\label{subsec:asym}
Roughly speaking, Propositions~\ref{prop:C1} and~\ref{prop:covent} imply that the evolution of asymmetry can be simulated by the evolution of entanglement. In particular, we can define asymmetry monotones for the states acting on~$\mathscr{H}$  in terms of the entanglement monotones of the states acting on~$\mathscr{W}$ to which they are mapped.  

\begin{definition}
\label{def:A}
For every bipartite entanglement monotone~$E$, we define the corresponding asymmetry monotone as,   
\begin{align}\label{ggg}
A^g_{E}: \mathcal{B}(\mathscr{H}) \rightarrow \mathbb{R}^+: \rho \mapsto E\left(\mathcal{C}_g (\rho)\right). 
\end{align}
\end{definition}
The following proposition ensures that $A^g_{E}$ is indeed an asymmetry monotone, assuming that $E$ is an entanglement monotone.
\begin{proposition}
Consider an entanglement monotone~$E$. If~$\rho \xrightarrow{\mathcal{E}_{\text{cov}}} \sigma$ is possible, then for every~$g \in G$, 
\begin{align}
\label{eq:entmon}
E\left(\mathcal{C}_g (\rho)\right) \geq E\left(\mathcal{C}_g (\sigma)\right). 
\end{align}
\end{proposition}
\begin{remark} 
Similar inequality holds in the case of non-deterministic~$G$-covariant CP-maps 
for the average of~$E$, assuming~$E$ is an ensemble monotone (See section~\ref{subsec:mon}). 
\end{remark}
\begin{proof}
The result follows directly from the definition ~\ref{def:isog} and the extension of proposition~\ref{prop:C1} to all isometries~$\mathcal{C}_g$.  
\end{proof}
 
As not all asymmetric states are taken to entangled states, the asymmetry monotone~$A^g_{E}$ is not faithful, even if~$E$ is itself a faithful entanglement monotone. However, proposition~\ref{prop:covent} allows us to define a faithful asymmetry monotone from the monotones~$A^g_{E}$: 
\begin{proposition}
\label{prop:sup}
The function
\begin{align}
A^{\text{sup}}_{E}:= \sup_{g \in G} A^g_{E},
\end{align}
where~$\sup_{g \in G}$ stands for the supremum taken over all~$g$ in~$G$, is a faithful asymmetry monotone,  provided~$E$ is a faithful entanglement monotone. 
\end{proposition} 
Replacing the supremum above by a maximum over the finite number of elements in $S\subset G$ (see Proposition~\ref{prop:covent}) will also lead to a faithful asymmetry monotone. For example, if $G=SU(2)$ then the function
$$
\max_{\hat{\text{n}} \in \{\hat{z},\;\hat{x}\}} A^{\hat{\text{n}}}_{E}
$$
is also an asymmetry monotone.
 
 \subsection{Unitary Transformation}
 If the CP-map is reversible, i.\ e.\, a unitary operation, then the condition of the monotonicity for the monotones~(\ref{ggg}) must be true in both directions, which in turn implies that the monotone functions must remain constant. 
 \begin{proposition}
 \label{prop:const}
Consider an entanglement monotone~$E$. If~$\rho \overset{{\mathcal{E}_{\text{cov}}} }\leftrightarrow \sigma$ is a reversible~$G$-covariant transformation, then for every~$g \in G$, $A^g_{E}$ is a conserved quantity; i.e. 
\begin{align}
\label{eq:entmon}
E\left(\mathcal{C}_g (\rho)\right)= E\left(\mathcal{C}_g (\sigma)\right). 
\end{align}
\end{proposition}
Thus, for closed systems governed by symmetric Hamiltonian, every entanglement monotone $E$ leads to new conserved quantities, $\{A^g_{E}\}_{g\in G}$. For a Hamiltonian that is symmetric with respect to the group $G$, the expectation values of the generators of $G$
are also conserved quantities. However, unlike $A^g_{E}$, for open systems these expectation values are not behaving monotonically.   

\section{Examples}
We now review in more detail  some examples of asymmetry monotones  that are constructed from entanglement  monotones through the class of LOCC simulating isometries.  

\subsection{The negativity of entanglement as a measure of Asymmetry} 

Many totally new asymmetry monotone can be constructed from entanglement monotones using the isometry~$\mathcal{C}$. Here we introduce two of such monotones for the first time.  
One such monotone uses the negativity of entanglement, and the other uses the logarithmic negativity~\cite{VW02, Plenio05}.  
 \begin{definition}
 The negativity of asymmetry is defined as, 
 \begin{align}
A_{N} (\rho):= \frac{\parallel\mathcal{C}(\rho)^{\Gamma} \parallel_{1}-1}{2}, 
 \end{align}
 and the logarithmic negativity of asymmetry is 
 \begin{align}
 A_{LN} (\rho):=\log \parallel\mathcal{C}(\rho)^{\Gamma} \parallel_{1},  
 \end{align}
 where $\Gamma$ denotes partial transpose and $\parallel \bullet \parallel_{1}$ is the 1-norm
 \begin{align}
 \parallel \rho\parallel_{1} = \text{Tr}\: \sqrt{\rho^{\dagger} \rho}.  
 \end{align}
 \end{definition}
Both negativity and logarithmic negativity are particular useful monotones as they are very easily computable for all states, pure or mixed. Note however that the negativity and the logarithmic negativity do not reduce to entropy functions for pure states.
For pure states, the negativity of asymmetry can be expressed in a very simple closed form.   
 As discussed in~\cite{GS08},  every pure state can be brought to a standard form with no explicit multiplicity index by~$G$-covariant transformations.
Consider the pure state in the standard form, 
$ 
|\psi\rangle =\sum_{j,m} \sqrt{p_{j,m}}\: |j; m\rangle.     
$ 
The norm of the partial transpose is 
 \begin{align}
 \parallel  \mathcal{C}&\left(|\psi\rangle\langle \psi|\right)^{\Gamma}  \parallel_1
 =\left(\sum_{j,m} \sqrt{p_{j,m}}\right)^2. 
 \end{align}   
 It follows that  the logarithmic negativity of asymmetry is equal to 
  \begin{align}
 A_{LN} \left(|\psi\rangle \langle \psi|\right) = 2 \log \left(\sum_{j,m} \sqrt{p_{j,m}}\right).  
 \end{align}
After simplifying the equations, the negativity of asymmetry can be expressed as
 \begin{align}
A_N \left(|\psi\rangle\langle\psi|\right) = \sum_{j\neq j',m\neq m'} \sqrt{p_{j,m}\:p_{j',m'}}.
 \end{align}


\subsection{Asymmetry Monotones Based On The Squashed Entanglement}
Squashed asymmetry is another new monotone constructed from the squashed entanglement monotone~\cite{CW04}. 
\begin{definition}
The squashed asymmetry is defined as
\begin{align}
A_{sq}(\rho):= E_{sq} \left(\mathcal{C}(\rho) \right),   
\end{align}
where, 
\begin{align}
E_{sq} \left(\mathcal{C}(\rho)\right)= \frac{1}{2} \inf_{C} S(A:B\:|\:C)
\end{align}
is the squashed entanglement.~$A$ and~$B$ denote the systems associated with the Hilbert spaces~$\mathscr{H}$ and~$\mathscr{H}_B$ respectively.~$C$ denotes an auxiliary system with Hilbert space~$\mathscr{H}_C$.   The minimum is taken over all extensions of $\mathcal{C}(\rho)$ to a tripartite state $\sigma^{ABC}$ acting on~$\mathscr{H}\otimes \mathscr{H}_B \otimes \mathscr{H}_C$, and the function~$S(A:B| C)$ is the conditional mutual entropy.  
\end{definition}
Squashed entanglement is known to be an additive monotone over the tensor product of states~\cite{CW04}. It is also a lower bound on entanglement of formation and an upper bound on the distillable entanglement. Its asymmetry counterpart introduced here could shed light on the properties of multiple copy~$G$-covariant transformations.  

\subsection{Measures Based On Distance}

Monotones based on how far states are from the set of non-resources are known as distance measures~\cite{VPRK97}. The geometric intuition  at play here can apply to various resources,  not just entanglement. If the resource is entanglement,  then the more entangled a state is, the further away it is from the set of separable states. The `distance' between any two states~$\rho$ and~$\sigma$ is measured by a function~$D(\rho, \sigma)$ with distance-like properties (e.g.  $D(\rho, \sigma)\geq 0$ with equality if and only if $\sigma=\rho$). The function~$D$, however, need not be literally a metric. All is needed is that $D$ preserve the partial order, and that~$D(\rho, \rho)= 0$ for all~$\rho$. $D$ need not satisfy the triangle inequality, for instance, and it need not even be symmetric. 
The distance-based monotone is defined as the minimum distance to the target set~Q:
\begin{align}
\label{ED}
E_D(\rho):= \inf_{\sigma \in \text{Q}} D(\rho, \sigma). 
\end{align} 

 In the case of entanglement, the target set is the set~SEP of separable states. If the function 
 $D(\rho,\sigma)=\text{Tr}\left[\rho\log\rho-\rho\log\sigma\right]$ is the relative entropy, then $E_D$ above is called the relative entropy of entanglement (REE).
 The REE has many nice properties and it plays a crucial role in the theory of entanglement~\cite{PV07,HHHH09}. 
 
 Just as in the previous subsection, we can use Eq.(\ref{ggg}) to define an asymmetry monotone that is based on the REE. 
 We call this monotone the relative entropy of asymmetry (REA). However,
 unlike the monotones in the previous subsection, distance based monotones of asymmetry can also be defined directly by choosing the target set~Q to be the set of~$G$-invariant states. In this case, if $D$ is taken to be the relative entropy then the resulting monotone is known to be the $G$-asymmetry~\cite{VAWJ08,GMS09}.
How the $G$-asymmetry is related the REA is an important question which we discuss here only partially. A more detailed study of the comparison is left for future work. 
  
 \subsubsection{The Relative Entropy of Asymmetry}

 As discussed above, an important and widely studied entanglement distance monotone is the REE, 
 \begin{align}
 E_R(\rho)= \min_{\sigma \in \text{SEP}} S\left(\rho \parallel \sigma\right),  
 \end{align}   
 where the relative entropy~
 $$
 S\left(\rho \parallel \sigma\right)=-S(\rho) - \text{Tr} \left(\rho \log \sigma\right),  
 $$
  is the distance function and where the infimum can be replaced with a minimum. The relative entropy is not symmetric and does not preserve the triangle inequality.    
   Following section~\ref{subsec:asym}, we can define a class of asymmetry monotones, 
 \begin{align}
 \label{def:ARg}
 A_{R}^{g}(\rho):= E_R\left(\mathcal{C}_g(\rho)\right), \;\; \forall g \in G. 
 \end{align}
Finally, we define the relative entropy of asymmetry (REA) to be the maximum monotone.  
\begin{definition}
\label{def:REA}
The relative entropy of asymmetry (REA) is the monotone,  
\begin{align}
  A_{R}^{\text{max}}(\rho):= \max_{s \in S} A_{R}^{s}(\rho), 
\end{align}
  where the finite subgroup~$S\subset G$  was defined by the property in Eq.~(\ref{eq:S}).   
\end{definition}
From the discussion in section~\ref{subsec:asym} it follows that~$A_{R}^{\max}$ is faithful, i.\ e.\  $A_{R}^{\max}(\rho) \equiv 0$ if and only if~$\rho$ is~$G$-invariant. 

\subsubsection{Comparison with~$G$-Asymmetry}

Choosing the set~INV of~$G$-invariant  as the target set~$Q$ for the states acting on~$\mathscr{H}$ leads to a measure known as the~$G$-asymmetry~\cite{VAWJ08} or, alternatively, the relative entropy of frameness~\cite{GMS09}, 
\begin{align}
 \label{def:AG}
 A_G:= \min_{\sigma \in \text{\text{INV}}} S\left(\rho \parallel \sigma\right) \equiv  S\left(\mathcal{G}(\rho)\right)-S(\rho).  
\end{align}
We refer to~$ A_G$ as~$G$-asymmetry for the rest of the paper. Here,~$\mathcal{G}(\rho)$ is the twirling operation discussed in Eq.~(\ref{def:twirl}) of section~\ref{subsec:G-cov}.  

In order to compare~$G$-asymmetry with REA,  let us first  consider a slightly different function, also based on the relative entropy of entanglement but with a different target set relative to which the distance is minimized.
\begin{figure}[h]
\centering
\includegraphics[width=0.5\textwidth]{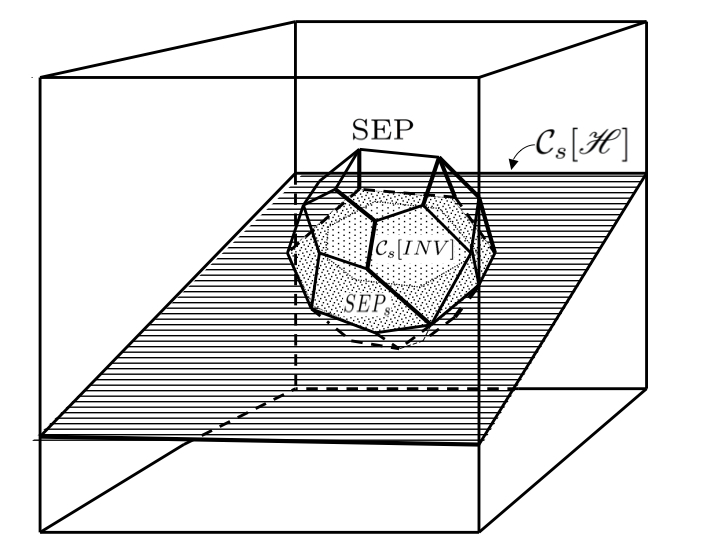}
\caption{A schematic depiction of the space of bipartite states.~$\text{SEP}_s$ is the intersection of the set  of separable states,~$\text{SEP}$, with the image of the~$\mathcal{C}_s$-isometry denoted here as~$\mathcal{C}_s [\mathscr{H}]$.  
The image of the set of~$G$-invariant states under~$\mathcal{C}_s$, denoted as~$\mathcal{C}_s [\text{INV}]$, is a strict subset of~$\text{SEP}_s$.} 
\label{fig:figSEP}
\end{figure}
 Each isometry~$\mathcal{C}_s$, for~$s \in S$,  leads in general to a strict distinct subset of SEP that act on~$\mathscr{H}\otimes \mathscr{H}_B$. 
 We denote the set by~$\text{SEP}_s$. That is,~$\text{SEP}_s$ is the intersection of SEP with the image of~$\mathcal{C}_s$ (See Figure~\ref{fig:figSEP}).   
We also denote the image of the set of~$G$-invariant states under~$\mathcal{C}_s$ as~$\mathcal{C}_s [\text{INV}]$. Note that if~$G$ is not Abelian, then~$\mathcal{C}_s [\text{INV}]$ is a strict subset of~$\text{SEP}_s$\footnote{If~$G$ is Abelian, then all separable states in~$\text{SEP}_s$ are images of invariant states and thus~$\mathcal{C}_s [\text{INV}]=\text{SEP}_s$ (See Appendix~\ref{A2}).}. For example, as we saw earlier,  
~$\text{SEP}_s$ also contains product states~$|\tilde{\phi}\rangle = |j, \lambda; m\rangle_s \otimes |m\rangle$  that are the images of the states~$ |j, \lambda; m\rangle_s$. Yet, the eigenstates~$ |j, \lambda; m\rangle_s$ are not~$G$-invariant when~$j\neq 0$. 
 We now define the function~$A_{R}^{*}$ as  
 \begin{align}
 \label{def:ARW}
 A_{R}^{s*}(\rho):= \min_{\sigma \in \mathcal{C}_s [\text{INV}] } S\left(\mathcal{C}_s(\rho) \parallel \sigma \right).  
\end{align}
The function~$A_{R}^{s*}$ can, in general, be greater than~$A_{R}^{\max}$ but it can never be  smaller. 
\begin{proposition}
\label{prop:AAs}
For every~$s \in S$, ~$A_{R}^{s*}$ are greater than or equal to the REA. 
 \begin{align}
 \label{def:ARAs} 
  A_{R}^{\max}(\rho) \leq  A_{R}^{s*}(\rho), \;\;\; \forall \rho \in \mathcal{B}(\mathscr{H}). 
 \end{align}
 \end{proposition}
 \begin{proof}
For any  given~$s \in S$,~$ \mathcal{C}_s [\text{INV}]   \subset \text{SEP}_{s} \subseteq \text{SEP}$.  
It follows that~$A_{R}^{s} \leq A_{R}^{s*}$, since~$A_{R}^{s}$ is obtained by minimizing the relative entropy over the larger set~$\text{SEP}$ that includes~$ \mathcal{C}_s [\text{INV}] $.  As this is true for all~$s \in S$, ~$A_{R}^{s*}$ is greater than or equal to the maximum ~$A_{R}^{\max}$ too.  
\end{proof}


The isomorphism between the two sets~INV and~$ \mathcal{C}_s [\text{INV}] $ implies that the minimum taken over~$ \mathcal{C}_s [\text{INV}]  $ in the definition of~$A_{R}^{s*}(\rho)$ coincides with the minimum of~$G$-asymmetry~$A_G$ in Eq.~(\ref{def:AG}). By this we mean that the separable state that minimizes the relative entropy in Eq.~(\ref{def:ARW}) is the image, under the isometry~$\mathcal{C}_s$, of the invariant state  that minimize the relative entropy in Eq.~(\ref{def:AG}).       
 
 To see this, consider the spectral decomposition~ of states~$\rho$ and~$\sigma$ acting on~$\mathscr{H}$, namely,~$\rho=\sum_i p_i |\psi_i\rangle \langle \psi_i|$ and~$\sigma=\sum_i q_i |\phi_i\rangle \langle \phi_i|$. Recall that~$\mathcal{C}_s$, being an isometry,  preserves the inner product between pure states\footnote{In fact,  as is apparent from definition~\ref{def:iso1},  the isometry~$\mathcal{C}_S$ merely `repeats' the weight label~$m$ for each eigenket~$|j, \lambda; m\rangle_S$ by attaching to it the ket~$|m\rangle$, i.\ e.\ ~$|j, \lambda; m\rangle_s \mapsto |j, \lambda; m\rangle_s \otimes |m\rangle$.}.         
It follows that the spectral decomposition of the image states 
 are ~$\mathcal{C}_s(\rho)=\sum_i p_i |\tilde{\psi}_i\rangle \langle \tilde{\psi}_i|$ and~$\mathcal{C}_s(\sigma)=\sum_i q_i |\tilde{\phi}_i\rangle \langle \tilde{\phi}_i|$, where~$|\tilde{\psi}\rangle$, and~$|\tilde{\phi}\rangle$ are themselves the images of~$|\psi\rangle$ and~$|\phi\rangle$,~i.\ e.\ , 
$
|\psi\rangle \xrightarrow{\mathcal{C})_s} |\tilde{\psi}\rangle, \;\; |\phi\rangle \xrightarrow{\mathcal{C}_s} |\tilde{\phi}\rangle.   
$    
Hence,  for every two states~$\rho$ and~$\sigma$, the two relative entropies~$S(\rho\parallel \sigma)$ and~$S\left(\mathcal{C}_s(\rho)\parallel \mathcal{C}_s(\sigma)\right)$ must be equal. Two corollaries follow:   

\begin{corollary}
\label{cor:ARWAG}
 For every~$s \in S$, the functions~$A_R^{s*}$ and~$A_G$ are identical,
~$
A_R^{s*}\equiv A_G. 
$
\end{corollary}
\begin{corollary}
\label{cor:ARsupAG}
 The~$G$-asymmetry is greater than or equal to the REA.  
 \begin{align}
 A_R^{\max} (\rho)  \leq A_G (\rho), \;\; \forall \rho \in \mathcal{B}(\mathscr{H}).
 \end{align} 
\end{corollary}
The relationship between~$G$-asymmetry and the REA goes deeper than what we have discussed so far, and  
our discussion here must be viewed only as an introductory treatment of the subject.  We leave a more complete discussion to future works.          

\section{Other Entanglement-Based  Selection Rules and Conservation Laws}
  
In this section, we consider a different isometry that has been used implicitly in the literature concerning symmetry and quantum reference frames~\cite{BRS07, GS08, GMS09}. 
The isometry is quite natural to consider, but as we will see, in general, it is not a LOCC-simulating isometry. Nevertheless, we will show that it still leads to new and independent necessary conditions for the manipulation of asymmetric states. 

We start by considering the Hilbert space decomposition of Eq.~(\ref{eq:Hj}). Irreps carrying subspaces~$\mathscr{H}_{j, \lambda}$ for fixed~$j$ are equivalent. Their direct sum, 
\begin{align}
\label{def:Hj}
\mathscr{H}_j:= \bigoplus_{\lambda} \mathscr{H}_{j,\lambda}
\end{align}
is isomorphic to~ 
$
\mathscr{H}_j \cong \mathscr{M}_j \otimes \mathscr{N}_j, 
$
where~$\mathscr{M}_j $ carries the~$j$'th irrep, and~$\mathscr{N}_j$ is the so called multiplicity space carrying the trivial representation of the group~\cite{BRS07}. It follows that~ 
$
\mathscr{H} \cong \mathscr{W}_{\mathcal{L}},
$
where, 
\begin{align}
\label{def:WL}
\mathscr{W}_{\mathcal{L}} := \bigoplus_j \mathscr{M}_j \otimes \mathscr{N}_j. 
\end{align}  
In~\cite{BRS07} the isomorphism of~$\mathscr{H}$ and~$\mathscr{W}_{\mathcal{L}}$ was assumed implicitly, but now we 
explicitly introduce the isometry connecting them.  
 
 \begin{definition}
\label{def:iso2} 
Let ~$\left\{ |j, m\rangle \right\}_m$ and~$\left\{ |j, \lambda \rangle \right\}_\lambda$ be the basis states spanning the spaces~$\mathscr{M}_j$ and~$\mathscr{N}_j$, respectively.  
Then~$\mathcal{L}: \mathcal{B}\left(\mathscr{H}\right) \rightarrow \mathcal{B}\left(\mathscr{W}_{\mathcal{L}} \right)$  
 is the isometry that maps,  
 \begin{align}
|j, \lambda; m\rangle \xrightarrow{\mathcal{L}}  |j, m\rangle \otimes |j,\lambda \rangle.  
 \end{align}
 \end{definition}  
Note that~$\mathscr{W}_{\mathcal{L}}\subset\mathscr{M} \otimes \mathscr{N}$, where
$
\mathscr{M}:=\bigoplus_j \mathscr{M}_j$ and~$\mathscr{N}:=\bigoplus_j \mathscr{N}_j
$.
Therefore, states in the image of $\mathcal{L}$ (i.e. states in $\mathscr{W}_{\mathcal{L}}$) can be viewed as bipartite states.
Moreover, if~$\rho$ is a~$G$-invariant state, then from Eq.~(\ref{eq:twirl2}) it follows that
 \begin{align}
 \mathcal{L}(\rho)=\sum_{j, \lambda} p_{j, \lambda} \left(\sum_m |j, m\rangle \langle j, m|\right)\otimes |j, \lambda\rangle\langle j,  \lambda|,  
 \end{align}
 which is a separable state~(see also~\cite{BRS07}). Similarly, any coherent superposition of states with different values of~$j$, 
 \begin{align}
 |\phi\rangle=\sum_{j,m,\lambda} c_{j, \lambda, m} |j, \lambda; m\rangle,   
 \end{align}
 is mapped to an entangled state, 
  \begin{align}
 |\tilde{\phi}\rangle=\sum_{j, m, \lambda} c_{j, \lambda, m} \:|j, m\rangle \otimes |j, \lambda\rangle.   
 \end{align}
Thus, ~$\mathcal{L}$ satisfies conditions $(2)$ and $(3)$ in Definition~\ref{def:lociso} of a LOCC-simulating isometry.  However, $\mathcal{L}$ is not a LOCC-simulating isometry since it does not in general satisfy condition~(1) of Definition~\ref{def:lociso},
as we show now for
the group~$G=SU(2)$. 

\subsection{$\mathcal{L}$ is not a LOCC-simulating isometry}
We now show that the entanglement of the bipartite states in the image of the isometry~$\mathcal{L}$ can in fact be increased by covariant transformations.   Consider the ${1/ 2}$-spin state~$\Psi=|\psi\rangle \langle \psi|$, where~$|\psi\rangle=|1/2; 1/2\rangle$, where we ignore the multiplicity index~$\lambda$, as it plays no role in what follows. Note that~$\mathcal{L}(\Psi)$ is a product state. Using Eq.~(\ref{eq:WE}), we see that the map~$\mathcal{E}_{1/2}$ takes~$\Psi$ to a state whose image is entangled. We only deal with fixed~$\alpha$ in~(\ref{eq:WE}), so we can remove it from our notation as well.  Consider the operator sum representation of the irreducible SU(2)-covariant map~$\mathcal{E}_{1/2}$ consisting of two Kraus operators~$K_{1/2, 1/2}$, and~$K_{1/2, -1/2}$. Because of the freedom in the choice of~$SU(2)$-covariant  Kraus operators,  we can choose them so that  they act on~$|\psi\rangle$ up to a normalization factor as     
\begin{align}
\label{eq:Kpsi}
K_{1/2, 1/2} |\psi\rangle &\propto |1; 1\rangle \xrightarrow{\mathcal{L}} |1;1\rangle \otimes|1\rangle ,  \nonumber \\
K_{1/2, -1/2} |\psi\rangle &\propto  |1;0\rangle +  |0;0\rangle \xrightarrow{\mathcal{L}}  |1;0\rangle\otimes|1\rangle +  |0;0\rangle\otimes |0\rangle.   
\end{align}
The state ~$\mathcal{L}\left(\mathcal{E}_{1/2}(\Psi)\right)$ is an equal mixture of the two states in the r.\ h.\ s.\ of Eq.~(\ref{eq:Kpsi}) and is thus an entangled state.     
It follows that the transformation    
\begin{align}
\mathcal{L}(\Psi) \mapsto \mathcal{L}\left(\mathcal{E}_{1/2}(\Psi)\right),
\end{align}
cannot be accomplished by LOCC.  

\subsection{Necessary conditions for the manipulation of asymmetric states}

Our motivation for introducing the isometries between the original  and the Kronecker product Hilbert spaces is to learn about~$G$-covariant  transformations. In particular, we study how the entanglement of the image states change.  In order to better understand how the entanglement changes under the isometry~$\mathcal{L}$, we now focus on the form of the maps  that act on the image states and mimic~$G$-covariant transformations.   The Wigner-Eckart theorem implies that, up to a projection to the subspace~$\mathscr{W}_{\mathcal{L}}$ of Eq.~(\ref{def:WL}), those are separable maps, i.\ e.\  of the form 
$$
\tilde{\mathcal{E}}_{\text{sep}}(\bullet)= \sum_x  \tilde{V}_x \otimes \tilde{K}_x (\bullet)  \tilde{V}^{\dagger}_x \otimes \tilde{K}^{\dagger}_x.  
$$

To see this, let~$\Pi_{\mathscr{W}_{\mathcal{L}}}$ denote the projection to the~$\mathscr{W}_{\mathcal{L}}$-space. As we saw in section~\ref{subsec:ito}, every~$G$-covariant transformation can be constructed from a set of irreducible tensor operators~$K_{J,M, \alpha}$. So we need only consider how~$K_{J,M, \alpha}$ are mimicked in the~$\mathscr{W}_{\mathcal{L}}$-space.     
If~$\rho$ is mapped to~$\sigma$ by~$K_{J,M, \alpha}$~($\sigma$ is in general subnormalized), then~$\mathcal{L}(\rho)$ is mapped to~$\mathcal{L}(\sigma)$ by the operator, 
\begin{align}
\label{eq:PWVK}
\tilde{K}_{J,M,\alpha}:= \tilde{V}_{J,M} \otimes \tilde{K}_{J,\alpha},  
\end{align}
followed by~$\Pi_{\mathscr{W}_{\mathcal{L}}}$.  The matrix elements of ~$\tilde{V}_{J,M}$ and~$\tilde{K}_{J,\alpha}$ are,  following the Wigner-Eckart theorem,   equal to the 
CG coefficient and the reduced matrix respectively,   
\begin{align}
\label{eq:CG}
&\langle j_2, m_2| \:\tilde{V}_{JM}\:|j_1,m_1\rangle= 
\left(
\begin{matrix}
  j_1 & J \\
  m_1 & M
 \end{matrix}
\right |
\left. 
 \begin{matrix}
 j_2  \\
  m_2
 \end{matrix}
 \right),
  \nonumber \\  
&\langle  j_2,\lambda_2 |\: \tilde{K}_{J,\alpha}\: |j_1,\lambda_1 \rangle =\; \langle j_2, \lambda_2\parallel K_{J,\alpha} \parallel j_1, \lambda_1\rangle.    
\end{align}   
  Again, here we consider only simply-reducible groups. For the generalization of the results of this section to all semi-simple Lie groups see Appendix~\ref{A1}.      

The entanglement of the image states can be increased only because of the projection~$\Pi_{\mathscr{W}_{\mathcal{L}}}$ in Eq.~(\ref{eq:PWVK}).  We can express the projection as~$\Pi_{\mathscr{W}_{\mathcal{L}}}=\sum_{j} \Pi_{j}$, where~
\begin{align}
\label{eq:pj}
\Pi_{j}=&\Pi_{\mathscr{M}_j} \otimes \Pi_{\mathscr{N}_j} \nonumber \\
:=& \sum_{m} |j,m\rangle \langle j,m| \otimes  \sum_{\lambda} |j, \lambda\rangle \langle j, \lambda|.    
\end{align}
Responsible for creating or increasing the entanglement are the cross terms ~$\Pi_j$ and~$\Pi_{j'}$ acting on both sides of~$\mathcal{L}(\rho)$ as 
\begin{align}
\label{eq:KW}
\mathcal{L}(\rho) \mapsto  \Pi_{\mathscr{W}_{\mathcal{L}}}\tilde{K}_{J,M,\alpha} \mathcal{L}(\rho) \; \tilde{K}^{\dagger}_{J,M,\alpha} \Pi_{\mathscr{W}_{\mathcal{L}}}.   
\end{align} 

In order to get rid of the cross terms, we proceed as follows:  Assume a given~$G$-covariant CP-map~$\mathcal{E}$ acting on~$\rho$, and the corresponding map on the bipartite state,  
\begin{align}
\label{eq:EPW}
 \tilde{\mathcal{E}} \left[\mathcal{L}(\rho)\right] = \Pi_\mathscr{W_{\mathcal{L}}} \left(\tilde{\mathcal{E}}_{\text{sep}} \left[\mathcal{L}(\rho)\right] \right) \Pi_\mathscr{W_{\mathcal{L}}},   
\end{align}
where~$\tilde{\mathcal{E}}_{\text{sep}}$ has an operator sum representation in terms of Kraus operators defined in Eq.~(\ref{eq:PWVK}).  
If, instead we consider the transformation 
\begin{align}
\label{eq:sigmabar}
\mathcal{L}(\rho) \mapsto \bar{\sigma}&=\sum_{j} \Pi_{j} \tilde{\mathcal{E}} \left[\mathcal{L}(\rho)\right]  \Pi_{j}\nonumber \\
&=\sum_{j} \Pi_{j} \left( \tilde{\mathcal{E}}_{\text{sep}} \left[\mathcal{L}(\rho)\right] \right)  \Pi_{j}. 
\end{align}
then the overall map remains a separable one. Note that the ~$\Pi_j$ are themselves separable. In fact, the transformation in~(\ref{eq:sigmabar}) can be implemented by LOCC. The reason is this: The superoperator~$\tilde{\mathcal{E}}_{\text{sep}}$  is comprised of operators~$\tilde{V}_{J,M} \otimes \tilde{K}_{J,\alpha}$. The projections~$ \Pi_{\mathscr{M}_j} \tilde{V}_{J,M}$  are unitary operators acting on the irrep-subspace~$\mathscr{M}_j$, as  their matrix elements are simply the CG-coefficients corresponding to a change of basis in~$\mathscr{M}_j$. Thus,  the whole transformation can be implemented by a series of local measurements by Alice, corresponding to operators~$\Pi_{\mathscr{N}_j} \tilde{K}_{J,\alpha}$, followed by the unitaries~$ \Pi_{\mathscr{M}_j} \tilde{V}_{J,M}$ performed by Bob.  

 It follows that the average entanglement of the state~$\tilde{\sigma}$ cannot exceed the entanglement of the initial state~$ \mathcal{L}(\rho)$. We state the result in the following proposition. 
\begin{proposition}
\label{prop:ave}
Let~$E$ be an ensemble entanglement monotone. We further assume that~$E$ is faithful and convex.  The ~$G$-covariant transformation~$\rho \mapsto \sigma$ is possible only if the following condition holds, 
\begin{align}
\label{eq:EE}
E\left(\mathcal{L}(\rho)\right) \geq E(\bar{\sigma}). 
\end{align}
\end{proposition}
\begin{proof}
The proposition is an immediate consequence of the fact that the transformation in~(\ref{eq:sigmabar}) can be implemented by LOCC.  
\end{proof}

Can we restate the condition of Eq.~(\ref{eq:EE})  in terms of new asymmetry monotones?  
Let us define the average initial state as, 
\begin{align}
\label{eq:rhobar}
\bar{\rho}:= \sum_{j} \Pi_{j} \mathcal{L}(\rho) \: \Pi_{j}. 
\end{align}
Clearly, the entanglement~$E(\mathcal{L}(\rho)) \geq E\left(\bar{\rho}\right)$.  
But does~$E\left(\bar{\rho}\right)$ exceed~$E\left(\bar{\sigma}\right)$ as well?  If this were true, then we could define an ensemble asymmetry monotone as~$A^{\text{ave}}_{E} (\rho):= E\left(\bar{\rho}\right)$ after all. However, this is not the case.  Consider the group~$G=SU(2)$, and let~$\rho=|\phi\rangle \langle \phi|$, where,  
$$
|\phi\rangle:= \frac{1}{\sqrt{2}} |3/2; 1/2\rangle + \frac{1}{\sqrt{2}} |1/2; 1/2\rangle.
$$ 
The image state, ~$\mathcal{L}(\rho)=|\tilde{\phi}\rangle \langle \tilde{\phi}|$, where, 
$$
|\tilde{\phi}\rangle:= \frac{1}{\sqrt{2}} |3/2; 1/2\rangle\otimes |3/2\rangle + \frac{1}{\sqrt{2}} |1/2; 1/2\rangle\otimes |1/2\rangle.
$$
is an entangled state.   
Also consider the irreducible~$SU(2)$-covariant CP-map~$\mathcal{E}_{1/2}$ (see section~\ref{subsec:ito}).  The state~$\bar{\rho}$ is a separable state, whereas the ensuing state~$\bar{\sigma}$ of Eq.~(\ref{eq:sigmabar}) is entangled. In other words, 
 $$
 E(\bar{\sigma}) \nleq E(\bar{\rho})=0.   
 $$   
Note that, in accordance with Proposition~\ref{prop:ave}, it is still true that~
$
0<E\left(\bar{\sigma}\right) \leq E\left(\mathcal{L}(\rho)\right).  
$  

In summary,  proposition~\ref{prop:ave} provides a necessary condition that all~$G$-covariant transformations must satisfy. Let us call such a necessary  condition a  \emph{general selection rule}. We have shown that the general selection rule in proposition~\ref{prop:ave} is not expressible in terms of asymmetry monotones of the initial and final states, but it is expressible in terms of the entanglement of their image states. This is an example of how asymmetry monotones are not the only relevant quantities in the study of symmetries of open systems.

\subsection{Conserved quantities}
 If we further restrict ourselves to \emph{reversible}~$G$-covariant transformations, still more interesting results can be deduced from the~$\mathcal{L}$-isometry. Unitary operations have only one Kraus operator. If~$G$ is non-Abelian, ~$G$-covariant unitaries exist only among~$G$-covariant transformations labeled by the identity representation,~$J=0$,  denoted by~$\mathcal{E}_{0,\alpha}=\mathcal{K}_{0,0,\alpha}$~(We consider the case of Abelian groups in Appendix~\ref{A2}.).  
 
 The unitary~$\mathcal{K}_{0,0,\alpha}$ maps each subspace~$\mathcal{H}_{j}$ in~(\ref{def:Hj})  to itself, and the corresponding bipartite operator~$\tilde{K}_{0,0,\alpha}$ has the form, 
 \begin{align}
 \tilde{K}_{0,0,\alpha}=\left(\sum_{j} \Pi_{\mathscr{M}_j}\right) \otimes \tilde{K}_{0,\alpha}.  
 \end{align}  
The above form is a direct consequence of the CG-coefficients in Eq.~(\ref{eq:CG}) for the case~$J=M=0$.  

Substituting for~$\mathcal{E}$ in Eq.~(\ref{eq:EPW}) shows that in this case the overall projection~$\Pi_{\mathscr{W}}$ into the subspace~$\mathscr{W}_{\mathcal{L}}$ can be dropped, because~$\tilde{K}_{0,0,\alpha}$ maps~$\mathscr{W_{\mathcal{L}}}$ to itself. Equivalently,~$\Pi_{\mathscr{W}_{\mathcal{L}}} \Pi_{j}=\Pi_{j}$, so that, 
$$
\Pi_{\mathscr{W}_{\mathcal{L}}} \tilde{K}_{0,0,\alpha} =\tilde{K}_{0,0,\alpha}.  
 $$
 The operator~$\tilde{K}_{0,0,\alpha}$ is of course a local unitary. It thus follows that for every reversible~$G$-covariant transformation~$\mathcal{E}$,  the entanglement of the image state in Eq.~(\ref{eq:EPW}) remains constant. In other words, we have identified a conserved quantity. 

\begin{proposition}
\label{prop:con}
For reversible~$G$-covariant transformations,~$\mathcal{E}_{0, \alpha}$, the function, 
\begin{align}
L(\rho):= E\left(\mathcal{L}({\rho})\right), 
\end{align} 
is a conserved quantity. 
\end{proposition}     
Hence, we have obtained new conservation laws for closed systems. The new conservation laws are not of the form of the expectation value of a generator of a Hamiltonian symmetry, but are instead in terms of entanglement monotones.  In the case of open systems and irreversible transformations, the conservation law is replaced with a general selection rule,   again in terms of entanglement monotones.

\section{Conclusion}
The present paper contains two major innovations: first, the notion of using of local operations to simulate symmetric dynamics, and  second,  the idea of applying  the well-known and well-studied  resource theory of entanglement to a totally different resource theory.  
Symmetric time evolutions described by covariant transformations are based on group structures, invariant subspaces and representation theory. It is not evident, at first,  that such structures have any connections to local operations and   tensor products of two or more systems. However, the link exists, and once found, is actually very simple. 
We found that the effect of an irreducible covariant operator on a ket~$|m\rangle$, labeled by the weight~$m$ of the algebra (ignoring the other labels),  is a simple translation by some fixed amount~$M$, ~$|m\rangle \rightarrow |m+M\rangle$. 
Thus, the local operators that simulate the~$G$-covariant transformations exploit a common feature of all Lie groups, i.\ e.\  how the weights are transformed. 

In this lies the strength of our method, as 
it applies equally to \emph{all Lie} groups and links them all to a sub-class of  local operations. In turn, this enables entanglement theory, as the resource theory of the restriction to LOCC, to be applied to the study of covariant transformations, irrespective of the symmetry group involved.   
Entanglement has been the focus of intense study and plays a central role in quantum information theory. This fact is reflected in the abundance of well investigated entanglement measures and monotones, each of which can now be used to extract information about the asymmetry of quantum states. One important consequence has been the realization that, for closed systems, entanglement serves as a conserved quantity, or a constant of motion.

There are various directions one can go from here.  First, we can ask what do entanglement considerations tell us about the specifics of~$G$-covariant transformations?  For example,  majorization of the Schmidt coefficients of the final pure state by the  coefficients of the initial state  is the necessary and sufficient condition for pure state to pure state transitions under LOCC.   If we apply the majorization condition to the images of the initial and final states for different isometries~$\mathcal{C}_g$, would we retrieve the exact form of the corresponding~$G$-covariant transformations? 


A second line of study concerns the case of finite groups.  
The isometries we introduced derive from the form of the Wigner-Eckart theorem for Lie groups.   
For finite groups, the form of the Wigner-Eckart theorem is  different and more complicated~\cite{Koster58}. 
If the finite-group version of the Wigner-Eckart theorem lends itself to the construction of LOCC-simulating isometries, then entanglement theory can be directly applied to finite groups as well. 

On a different note, we have not considered the case of many-copy transformations and asymptotic limits in the present paper. Many questions of interest can be asked in this respect, including additivity of the measure and possible applications to the problem of distillation of asymmetry resources. 

Finally, a fourth direction for future research  suggested by our result is to look for similar conditions in other resource theories.  For example, the restriction to Gaussian operations  results in a new resource theory where non-Gaussian states are resources~\cite{GC02}. Another example is thermodynamics. Thermodynamics has been recognized as an energy preserving resource theory where transformations are restricted to operations that do not increase the total energy~\cite{HO11}, and already, connections between thermodynamics, viewed as a resource theory, and entanglement have been demonstrated~\cite{HO02, H08, HO11}.  If the restricted set of operations in any of those resource theories are simulated by local operations, then it would be possible to employ entanglement theory to the study of those resource theories as well.

\section{acknowledgments}   
We appreciate valuable discussions with Varun Narasimhachar and Iman Marvian. BT also appreciates valuable discussions with P.\  S.\ Turner. This research has been supported by Alberta Innovates,   
the Natural Sciences and Engineering Research Council,
General Dynamics Canada,
and the Canadian Centre of Excellence for Mathematics of Information Technology and Complex Systems (MITACS).

\appendix
\section{Generalized Wigner-Eckart Theorem In The Presence Of Outer Multiplicities}
\label{A1}
The main results of the paper can be extended to the general case where the Kronecker product of the algebra associated with the group is not simply reducible. An algebra~$H$ is not simply reducible when the algebra has outer  multiplicities, i.\ e.\ multiplicities arising due to the coupling of the irreps. 
We now consider the general form of the Wigner-Eckart theorem,     
\begin{align}
\label{eq:A-WE}
\langle  j',\lambda'; m' |K_{J,M,\alpha}& |j, \lambda; m \rangle= \nonumber \\ 
&\sum_{\mu} \left(
 \begin{matrix}
  j & J \\
  m & M
 \end{matrix}
\right |
\left. 
 \begin{matrix}
 j',   \mu  \\
  m'
 \end{matrix}
 \right)
\: \langle j', \lambda' \parallel K_{J, \alpha} \parallel j, \lambda\rangle_{\mu}, 
\end{align}
where~$\mu$ is the outer multiplicity index for the irrep~$[j']$ due to the coupling~
$
[j] \otimes [J] \mapsto [j'].  
$ 
Here, we have used the symbol~$[j]$ to denote the representation labeled by~$j$, and similarly for other representations.  The terms~
$
\left(
\begin{matrix}
  j & J \\
  m & M
 \end{matrix}
\right |
\left. 
 \begin{matrix}
 j',   \mu  \\
  m'
 \end{matrix}
 \right)
$   
are the general Clebsch-Gordan coefficients, depending, in the general case, on the outer multiplicity~$\mu$ in addition to the irrep and weight labels.

If the transformation~$K_{J,M,\alpha}$ is unitary,~$J$ and~$M$ still remain the labels of the identity representation,~$J=M=0$. Coupling to the identity representation never results in outer multiplicities. Thus, the results for~$G$-covariant unitaries in the paper is valid for the general case.

\subsection{The Set of Isometries~$\left\{\mathcal{C}_g \right\}$}

All the Clebsch-Gordan coefficients are identically zero unless, as before, the weights labelling the bra and the ket, and the tensor operator satisfy the relation, 
$$
m+M=m'. 
$$
It follows that, as far as the weights are concerned, the same translation operator as in Eq.~(\ref{def:T}) applies to all the terms in the r.\ h.\ s. of~(\ref{eq:A-WE}), and thus the same set of isometries~$\mathcal{C}$ and~$\mathcal{C}_g$ in the definitions~\ref{def:iso1} and~\ref{def:isog} of section~\ref{sec:sim} still satisfy all the conditions of a LOCC-simulating isometry of definition~\ref{def:lociso}. 

\subsection{The Isometry~$\mathcal{L}$}

The situation is more complicated for the isometry~$\mathcal{L}$. The existence of outer multiplicities implies that we must define new Hilbert spaces to embed the original Hilbert space, i.\ e.\ Hilbert spaces that include the outer multiplicities in the label of their basis states. Let~
$$
\mathscr{M}=\text{span} \left\{ |j,\mu ;m\rangle\right\}_{j,\mu, m}, 
$$ 
be the space spanned by the basis states~$ |j,\mu ;m\rangle$. Here,~$j$ and~$m$ are, as before,  the irrep label and the weight label respectively. We have included an additional label~$\mu$, ranging over~$\mu=0\cdots \infty$, that we will shortly relate to the outer multiplicities. Similarly, let
$$
\mathscr{N}=\text{span} \left\{ |j,\mu ; \lambda\rangle\right\}_{j,\mu, \lambda}, 
$$
where~$\lambda$ is the label for the (initial) irrep multiplicities.   Also, define~$\mathscr{M}_j=\text{span} \left\{|j, 0; m \rangle\right\}_{m}$ and~$\mathscr{N}_j=\text{span} \left\{|j, 0; \lambda \rangle\right\}_{\lambda}$, and  
$$
\mathscr{W}_{\mathcal{L}}:= \bigoplus_{j} \mathscr{M}_j \otimes \mathscr{N}_j. 
$$
 As before, we can define the isometry~$\mathcal{L}$ by specifying how it acts on the basis states. 

\begin{definition}
\label{def:isoLA} 
$\mathcal{L}: \mathcal{B}\left(\mathscr{H}\right) \rightarrow \mathcal{B}\left(\mathscr{W}_{\mathcal{L}} \right)$  
 is the isometry that maps,  
 \begin{align}
|j, \lambda; m\rangle \xrightarrow{\mathcal{L}}  |j,0; m\rangle \otimes |j,0; \lambda \rangle.  
 \end{align}
 \end{definition}  
Clearly,~$\mathscr{W}_{\mathcal{L}} \subset \mathscr{M} \otimes \mathscr{N}$, and thus the states in the image of $\mathcal{L}$ (i.e. states in $\mathscr{W}_{\mathcal{L}}$) are  bipartite states.
Let~$K_{J,M,\alpha}$ be an irreducible~$G$-covariant operator~(see Eq.~(\ref{def:irtenop}) in section~\ref{subsec:ito}). The operator acting on~$\mathscr{M}\otimes \mathscr{N}$ that mimics~$K_{J,M,\alpha}$ can again be expressed as a separable state followed by a projection to the image subspace~$\mathscr{W}_{\mathcal{L}}$.   
Assume~$\rho$ is mapped to (a in general subnormalized)~$\sigma$ by~$K_{J,M, \alpha}$.~$\mathcal{L}(\rho)$ is then mapped to~$\mathcal{L}(\sigma)$ by the operator, 
\begin{align}
\label{eq:PWVK-A1}
\tilde{K}_{J,M,\alpha}:= \tilde{V}_{J,M} \otimes \tilde{K}_{J,\alpha},  
\end{align}
followed by~$\Pi_{\mathscr{W}_{\mathcal{L}}}$.  The general form of the Wigner-Eckart theorem~(\ref{eq:A-WE}) implies,   
\begin{align}
\label{eq:CG}
&\langle j_2, \mu_2; m_2| \:\tilde{V}_{JM}\:|j_1, \mu_1; m_1\rangle= \: 
\left(
\begin{matrix}
  j_1 & J \\
  m_1 & M
 \end{matrix}
\right |
\left. 
 \begin{matrix}
 j_2,   \mu_2  \\
  m_2
 \end{matrix}
 \right)
, \nonumber \\  
&\langle  j_2,\mu_2; \lambda_2 |\: \tilde{K}_{J,\alpha}\: |j_1,\mu_1; \lambda_1 \rangle =\; \langle j_2, \lambda_2\parallel K_{J,\alpha} \parallel j_1, \lambda_1\rangle_{\mu_2}.   
\end{align}   
Note that  r.\ h.\ s.\  does not depend on the value of~$\mu_1$ in either equation. 
The map~$\Pi_{\mathscr{W}_{\mathcal{L}}}$ is $\Pi_{\mathscr{W}_{\mathcal{L}}}=\sum_{j, \mu} \Pi_{j, \mu}$, where~
\begin{align}
\label{eq:pj}
\Pi_{j, \mu}:=&\Pi_{\mathscr{M}_j} \otimes \Pi_{\mathscr{N}_j} \nonumber \\
:=& \sum_{m} |j,0;m\rangle \langle j,\mu;m| \otimes  \sum_{\lambda} |j, 0; \lambda\rangle \langle j, \mu; \lambda|.    
\end{align}
 for a given~$G$-covariant CP-map~$\mathcal{E}$ acting on~$\rho$, the corresponding map on the bipartite state is 
\begin{align}
\label{eq:EPW-A}
 \tilde{\mathcal{E}} \left[\mathcal{L}(\rho)\right] = \Pi_\mathscr{W_{\mathcal{L}}} \left(\tilde{\mathcal{E}}_{\text{sep}} \left[\mathcal{L}(\rho)\right] \right) \Pi_\mathscr{W_{\mathcal{L}}},   
\end{align}
where~$\tilde{\mathcal{E}}_{\text{sep}}$ has an operator sum representation in terms of Kraus operators defined in Eq.~(\ref{eq:PWVK}).  
$\tilde{\mathcal{E}} \left[\mathcal{L}(\rho)\right]$ is still not a LOCC-simulating CP-map. The cross terms~$\Pi_{j,\mu}$ and~$\Pi_{j',\mu'}$ acting on both sides of~$\mathcal{L}(\rho)$ render the overall CP-map a non-separable one. 
However, we can destroy the cross terms here too,  by applying the set of projections~$\Pi_{j, mu}$ separately on both sides and then taking the average of the maps,  as follows, 
\begin{align}
\label{eq:sigmabar-A}
\mathcal{L}(\rho) \mapsto \bar{\sigma}&=\sum_{j, \mu} \Pi_{j, \mu} \: \tilde{\mathcal{E}} \left[\mathcal{L}(\rho)\right] \:  \Pi_{j,\mu}\nonumber \\
&=\sum_{j,\mu} \Pi_{j,\mu} \left( \tilde{\mathcal{E}}_{\text{sep}} \left[\mathcal{L}(\rho)\right] \right)  \Pi_{j,\mu}. 
\end{align}
The overall map is separable now, and the condition~$E(\rho) \geq E(\bar{\sigma})$ must hold if the transition from~$\rho$ to~$\sigma$ is possible.

\section{Abelian Lie Groups}
\label{A2} 
The irreducible representations of Abelian groups are $1$-dimensional. The irrep label is always the highest weight and  $1$-dimensional  irreps have only one weight. Thus, the irrep label and the weight label are the same.  We use the label~$n$ for the irreps of an Abelian group, and to conform to the notation of the rest of the paper, we label the basis states as~$|n, \lambda; n\rangle$. We presently show that the results of the paper are greatly simplified in the case of Abelian groups. In particular, we show that the isometries~$\mathcal{C}_g$ are all equivalent with each other, and are furthermore equivalent to the isometry~$\mathcal{L}$. 

\begin{definition}
\label{def:isoCB} 
$\mathcal{C}: \mathcal{B}\left(\mathscr{H}\right) \rightarrow \mathcal{B}\left(\mathscr{W}_{\mathcal{L}} \right)$  
 is the isometry that maps,  
 \begin{align}
|n, \lambda; n\rangle \xrightarrow{\mathcal{C}}  |n,\lambda; n\rangle \otimes |n\rangle.  
 \end{align}
 \end{definition}  

\begin{definition}
\label{def:isoLB} 
$\mathcal{L}: \mathcal{B}\left(\mathscr{H}\right) \rightarrow \mathcal{B}\left(\mathscr{W}_{\mathcal{L}} \right)$  
 is the isometry that maps,  
 \begin{align}
|n, \lambda; n\rangle \xrightarrow{\mathcal{L}}  |n; n\rangle \otimes |n; \lambda \rangle.  
 \end{align}
 \end{definition}  
 
First, note that the action of a group element on the basis kets is to merely add a phase, 
$$
U(g)|n, n, \lambda\rangle = e^{\imath \theta_{g,n}} |n, n, \lambda\rangle.  
$$
Thus, the definition~\ref{def:isog} of~$\mathcal{C}_g$ implies,  
\begin{align}
\mathcal{C}_g \equiv \mathcal{C},\;\; \forall g \in G. 
\end{align}
 
The form of irreducible~$G$-covariant transformation is also simplified to, 
\begin{align}
\label{eq:WE-B}
\langle  n',\lambda'; n' |K_{N,N, \alpha}& |n, \lambda; n \rangle=
\delta_{n',n+N} \langle n', \lambda' \parallel K_{N, \alpha} \parallel n, \lambda\rangle, 
\end{align}
 or equivalently 
\begin{align}
\label{eq:KN-B}
	K_{N,N,\alpha}=\sum_n c^{(N, \alpha)}_{n, \lambda, \lambda'} \:|n+N, \lambda'; n+N \rangle\langle n, \lambda; n|,  
\end{align}
where~$ c^{(N, \alpha)}_{n, \lambda, \lambda'} = \langle n', \lambda' \parallel K_{N, \alpha} \parallel n, \lambda\rangle$. 

Assume~$\rho$ is mapped to (a in general subnormalized state)~$\sigma$ by~$K_{J,M, \alpha}$.~$\mathcal{C}(\rho)$ is then mapped to~$\mathcal{C}(\sigma)$ by the operator, 
\begin{align}
\label{eq:CK-B1}
\tilde{K}^{\mathcal{C}}_{N,N,\alpha}= K_{N,N,\alpha} \otimes \sum_{n} |n+N\rangle \langle n|.   
\end{align}

Equivalently,~$\mathcal{L}(\rho)$ is mapped to~$\mathcal{L}(\sigma)$ by the operator, 
\begin{align}
\label{eq:LK-B1}
&\tilde{K}^{\mathcal{L}}_{N,N,\alpha}= \nonumber \\
& \sum_{n} |n+N; n+N\rangle \langle n; n| \otimes \sum_n c^{(N, \alpha)}_{n, \lambda, \lambda'} \:|n+N; \lambda' \rangle\langle n; \lambda|. 
\end{align}

$\tilde{K}^{\mathcal{C}}_{N,N,\alpha}$ can be implemented by a LOCC-transformation, as~$\mathcal{C}$ is a  LOCC-simulating isometry. Now, interestingly, the simulating operator of the second isometry,~$\tilde{K}^{\mathcal{L}}_{N,N,\alpha}$ is implementable by LOCC-transformations as well. So in the case of the Abelian groups, the isometry~$\mathcal{L}$ is also a LOCC-simulating isometry. In fact, the forms of~$\tilde{K}^{\mathcal{C}}_{N,N,\alpha}$ and~$\tilde{K}^{\mathcal{L}}_{N,N,\alpha}$ are similar, both comprised of the tensor product of a copy of the original~$G$-covariant operator~$K_{N,N,\alpha}$ and  a translation operator, and the isometry 
$$
|n, \lambda; n\rangle \otimes |n\rangle \mapsto |n; n\rangle \otimes |n; \lambda\rangle 
$$
maps one set of LOCC-transformations to an equivalent set of  LOCC transformations. In this sense, the two isomtries~$\mathcal{C}$ and $\mathcal{L}$ are equivalent. 

The image state under either isometry is an entangled state if and only if the initial state has no coherence in~$n$, i.\ e.\ if the state is a coherent superposition of states with different values of~$n$. States acting on the original Hilbert space~$\mathscr{H}$  with no coherence in~$n$ are the~$G$-invariant states, as the twirling operations destroys the coherence in~$n$. 
\begin{proposition}
\label{prop:G-inv-B1}
If~$G$ is an Abelian group, then the image state~$\mathcal{C}(\rho)$ (or equivalently~$\mathcal{L}(\rho)$) is a separable state if and only if the initial state~$\rho \in \mathcal{B}(\mathscr{H})$ is~$G$-invariant.  
\end{proposition}

Finally, as a corollary we note that the `average state', ~$\bar{\sigma}$ of Eq.~(\ref{eq:sigmabar}), is always a separable state and has no entanglement.

\end{document}